\documentclass[11pt,twoside]{article}
\usepackage[utf8]{inputenc}
\usepackage[T1]{fontenc}
\usepackage[left=1in,top=1in,right=1in]{geometry}
\usepackage{calc}
\usepackage{amsmath}
\usepackage{amssymb}
\usepackage{amsfonts}
\usepackage{amsthm}
\usepackage{epsfig}
\usepackage{color}
\usepackage{psfrag}
\usepackage{enumerate}
\usepackage{graphicx}
\usepackage{caption}
\usepackage{subcaption}
\usepackage[round]{natbib} 
\usepackage[linesnumbered,boxed]{algorithm2e}
\usepackage{fullpage}
\usepackage{tikz}
\usetikzlibrary{plotmarks,positioning,shapes,arrows,backgrounds,patterns}

\usepackage{hyperref}
\hypersetup{
    unicode=false,          
    pdftoolbar=true,        
    pdfmenubar=true,        
    pdffitwindow=false,     
    pdfstartview={FitH},    
    colorlinks=true,       
    linkcolor=red,          
    citecolor=blue,        
    filecolor=magenta,      
    urlcolor=green           
}

\newcommand{\tr}{\text{tr}}
\DeclareMathOperator{\atan}{\text{atan}}
\def\cov{\text{cov}}
\def\var{\text{var}}
\def\E{\text{E}}
\newtheorem{lemma}{Lemma}[section]

\newtheorem{theorem}{Theorem}[section]
\newtheorem{proposition}{Proposition}[section]

\numberwithin{equation}{section}

\begin{document}

\author{Matthieu Wilhelm ${}^{1}$\footnote{Corresponding author. E-mail: matthieu.wilhelm@unine.ch
.} \ and \ Laura M. Sangalli${}^{2}$ \\
{}\\
${}^{1}$ \ \normalsize Institut de Statistique\\[-0.1cm]
\normalsize Universit\'e de Neuch\^atel\\[-0.1cm]
{}\\
${}^{2}$ \ \normalsize MOX -- Laboratorio di Modellistica e di Calcolo Scientifico \\[-0.1cm]
\normalsize Politecnico di Milano\\[-0.1cm]
{}\\
}

\title{\textbf{Generalized Spatial Regression with Differential Regularization}}

\maketitle


\begin{abstract}
We aim at analyzing geostatistical and areal data observed over irregularly shaped spatial domains and having a distribution within the exponential family. We propose a generalized additive model that allows to account for spatially-varying covariate information. The model is fitted by maximizing a penalized log-likelihood function, with a roughness penalty term that involves a differential quantity of the spatial field, computed over the domain of interest. Efficient estimation of the spatial field is achieved resorting to the finite element method, which provides a basis for piecewise polynomial surfaces. The proposed model is illustrated by an application to the study of criminality in the city of Portland, Oregon, USA.
\end{abstract}


\section{Introduction and motivation}

We propose a generalized regression model for spatially distributed data, when the response variable has a distribution within the exponential family. One of the main features of the model is that it is able to deal with domains having complex shapes, characterized for instance by strong concavities or holes, and where the shape of the domain influences the behavior of the phenomenon.
To illustrate this problem, we consider the study of criminality over the city of Portland, Oregon, USA.
The left panel of Figure \ref{gr:port_data_and_tri} shows a map of this city, cut in two parts by  the Willamette river. The two parts of the city are connected only by a few bridges. The dots over the map indicates the locations of all the crimes reported in 2012. It is apparent that the variation of the phenomenon is not smooth across the river. The map also shows the municipality districts. Census information is available for each district, such as the total number of residents per district. We would like to study the spatially-varying  criminality in the city, taking into account the auxiliary information based on the census. Since the covariate  is available at the level of districts, we aggregate also the crimes, thus considering as outcome of interest the total  crime count over each districts.

When analyzing these data, it appears crucial to accurately take into account the shape of the domain. Features such as the river and the bridges in fact influence the phenomenon expression; see, e.g., \cite{Chainey-Ratcliffe-2005,Ratcliffe-2010}.  See also \cite{Bernasco-Elffers-2010} for a comprehensive review on the statistical analysis of spatial crime data.
Moreover, not restricted to criminology applications, there is a vast literature devoted to the study of spatially varying data having a distribution within the exponential family \citep[see, e.g.,][and references therein]{Diggle07}.
However, these methods are not well suited for the analysis of the data here presented, as they do not account for the complex shape of the problem domain, neglecting for instance natural barriers such as the river.

Recently, some spatial data analysis methods have been proposed where the shape of the domain is  directly specified in the model; these include the spatial regression models with differential regularization proposed in \cite{Ramsay02} and \cite{Sangalli13}, and the soap film smoothing introduced by \cite{Wood08}. Here we propose an extension of the methodology presented in \cite{Sangalli13}, allowing to model response variables having a distribution within the exponential family, including binomial, gamma and Poisson outcomes.
Specifically, we maximize a penalized log-likelihood function with a roughness penalty term that involves a differential quantity  of the spatial field computed over the domain of interest. We name the resulting method GSR-PDE: Generalized Spatial Regression with PDE penalization. To solve the estimation problem, we derive a functional version of the Penalized Iterative Reweighted Least Squares (PIRLS) algorithm \citep{Osullivan86}.  This functional version of the  PIRLS algorithm can be used to maximize penalized log-likelihoods with general quadratic  penalties involving a functional parameter. Likewise  \cite{Ramsay02} and \cite{Sangalli13}, the proposed models make use of  finite elements over a triangulation of the domain of interest, to obtain accurate estimates of the spatial field.
See \cite{Ramsay00} for an earlier use of finite elements in a spatial data analysis context, and \cite{Lindgren11} for the purpose of fitting Gaussian random fields. Domain  triangulations are able to efficiently describe domains with complex geometries. The right panel of Figure \ref{gr:port_data_and_tri} shows a triangulation of the city of Portland. The triangulation accurately renders the strong concavities in the domain represented by the river, and also very localized and detailed structures of the domain such as the bridges that connect the two parts of the city center. The proposed model is detailed both for the case of geostatistical data and for the case of areal data. The model versions for geostatistical  and for areal data are special cases of a unique model, although for simplicity of exposition we introduce first the version for geostatistical data, then the one for areal data, and we postpone to the appendix the unified modelling formulation. Some comparative simulation studies show the good performances of the model.

The paper is organized as follows. In section \ref{sec::model}, we introduce the model, detailed in the case of geostatistical data. In the section \ref{sec::just_PIRLS_variational}, we derive the functional version of the PIRLS algorithm. In section \ref{sec::fem}, we describe the numerical implementation of the fitting procedure. In Section \ref{sec::areal_data_model} we specify the model version for areal data.
Section \ref{sec::sim} is devoted to simulation studies and Section \ref{sec::application} to the study of criminality over the city of Portland. Finally, Section \ref{sec:discussions} draws some directions for future research. All technical details and proofs are deferred to the Appendix.

\begin{figure}[htbp]
\begin{subfigure}{.49\linewidth}
\includegraphics[scale=0.159]{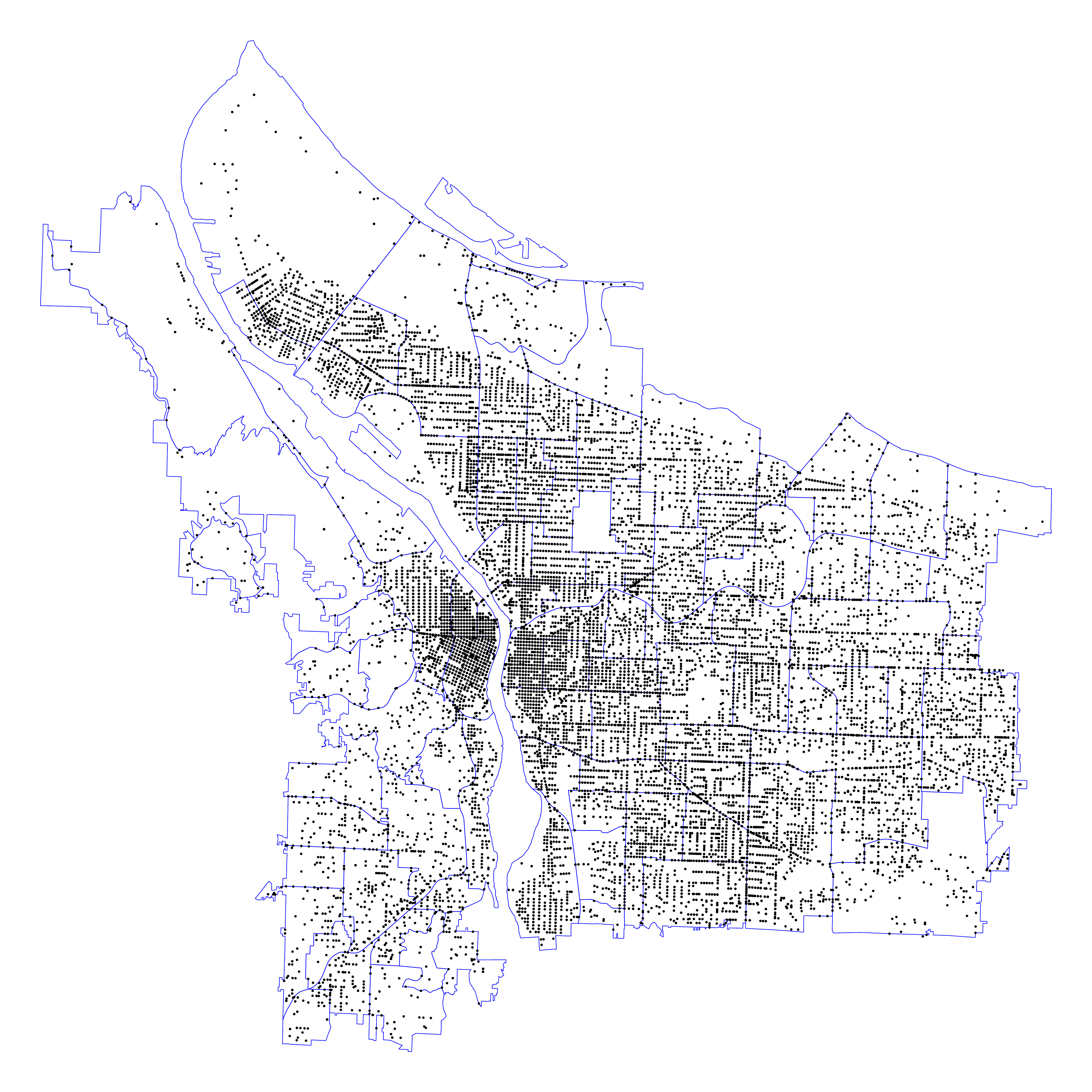}
\end{subfigure}
\begin{subfigure}{.49\linewidth}
\includegraphics[scale=0.159]{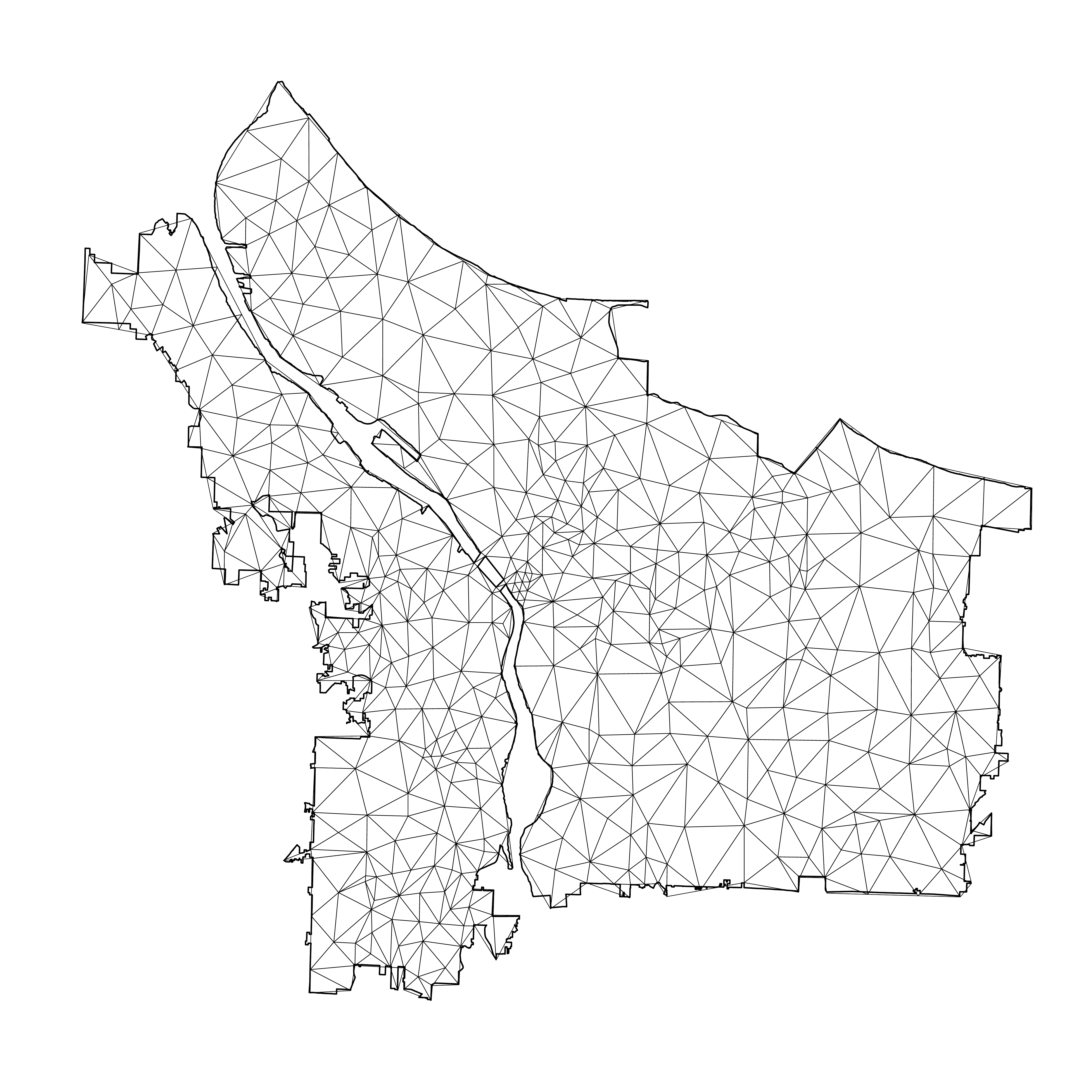}
\end{subfigure}
\caption{\label{gr:port_data_and_tri} Left: crime locations in the city of Portland, Oregon, in 2012. Right: the triangulation of the domain, with the borders of the city districts highlighted in blue. In the data analysis, crimes are aggregated over districts, leading to the total crime counts shown in the left panel of Figure \ref{gr::crime_dist_Port}.}
\end{figure}

\section{Model version for geostatistical data}
\label{sec::model}

We consider a bounded domain $\Omega\subset \mathbb{R}^2$ with a regular boundary $\partial \Omega\in {C}^2(\mathbb{R}^2)$. We consider $n$ fixed locations $\mathbf{p}_1,\dots,\mathbf{p}_n \in \Omega,$ where $\mathbf{p}_i=(p_{1i},p_{2i})$. At each  $\mathbf{p}_i$ we observe the realization $y_i$ of a real variable of interest $Y_i,$ and a vector of covariate information $\mathbf{x}_i\in\mathbb{R}^q$. We assume $Y_1,\dots,Y_n$ are independent, with $Y_i$  having a distribution within the exponential family, with mean $\mu_i$ and common scale paramenter $\phi$. We  model $\mu_i$ by the following generalized additive model: \begin{equation}
\label{equ::model_intro}
g(\mu_i) = \theta_i=  x_i^t \boldsymbol{\beta} +{f}(\mathbf{p}_i),
\end{equation}
where $g$ is a continuously differentiable and strictly monotone canonical link function,  $\boldsymbol{\beta} \in \mathbb{R}^q$ is a vector of coefficients, and $f$ is a smooth field over $\Omega,$  lying in a suitable functional space $\mathcal{F}$. The parameter $\theta$ is referred to as the canonical parameter.

We then propose to estimate the regression coefficients $\boldsymbol{\beta} \in \mathbb{R}^q$   and the spatial field $f\in\mathcal{F}$ by maximizing a penalized log-likelihood functional:
\begin{equation}
\label{equ::pen_like}
\mathcal{L}_{p}(\boldsymbol{\beta}, f)=\sum_{i=1}^n l(y_i;\theta_i(\boldsymbol{\beta}, f))-\lambda \int_{\Omega} \big(\Delta f(\mathbf{p})\big)^2d\mathbf{p},
\end{equation}
where $l(\cdot)$ is the log-likelihood and $\theta_i(\boldsymbol{\beta}, f)=\mathbf{x}_i^t \boldsymbol{\beta} +
f(\mathbf{p}_i)$. Here $\lambda$ is a positive smoothing parameter and the Laplacian $\Delta f=\partial^2 f/\partial p_1^2+ \partial^2 f/\partial p_2^2$ is a  measure of the local curvature of the field $f.$  The higher $\lambda$ is, the more we control the wiggliness of the spatial field $f$, the smaller $\lambda$ is, the more we allow flexibility of $f$. As discussed in Section \ref{sec:discussions}, more complex roughness penalties may be considered. \cite{Azzimonti1,Azzimonti2} for instance  show that by changing the regularizing terms
and considering more complex differential operators it is
possible to include in the model a priori information about
the spatial variation of the phenomenon under study, and model also
space anisotropies and non-stationarities. As commented in Appendix \ref{app::spatial-variation}, the regularizing term
in (\ref{equ::pen_like}) effectively induces the spatial variation structure of the estimator
and different regularizations  imply different variation structures.

  In the case of Gaussian observations, with mean $\mu_i=\theta_i$ and constant variance $\sigma^2$, the maximization of the penalized log-likelihood function is
  equivalent to the minimization of the penalized least-square functional considered in \cite{Sangalli13}. In this case, the quadratic form of the functional  allows to characterize analytically the minimum of the penalized least square functional (or equivalently, the maximum of the penalized log-likelihood functional), and thus to characterize the estimators $\hat{\boldsymbol{\beta}} \in \mathbb{R}^q$   and   $\hat{f}\in\mathcal{F}.$ Outside of the Gaussian case, it is not possible to characterize analytically the solution of the estimation problem. On the other hand, we cannot either apply the standard PIRLS algorithm, developed by \cite{Osullivan86} for the maximization of a penalized log-likelihood functional in the context of generalized additive models. This is due to the fact that  the penalized log-likelihood in (\ref{equ::pen_like})  involves a function parameter, the spatial field $f$, and the maximization is over the space $\mathbb{R}^q \times \mathcal{F},$ where $\mathcal{F}$ is an infinite-dimensional space.  In the following section, we thus present a functional version of the PIRLS algorithm, that can be used to find an approximate solution of the estimation problem here considered. More generally, the proposed functional version of the PIRLS algorithm can be employed   in the context of generalized linear models with a functional parameter, to maximize a penalized log-likelihood that has a quadratic penalty on the functional parameter.

\section{Functional version of the PIRLS algorithm}
\label{sec::just_PIRLS_variational}

We consider the following parametrization of a distribution from the exponential family:$$f_{Y}(y; \theta, \phi)=\exp\left\{(y\theta-b(\theta))/a(\phi) +c(\phi, y)\right\},$$
where $a(\cdot)$, $b(\cdot)$ and $c(\cdot)$ are functions subject to some regularity constraints \citep[see, e.g.,][]{Nelder}. For sake of simplicity, we only consider canonical link functions, that is $b'(\theta)= g^{-1}(\theta),$ and we make no distinction between the natural and the canonical parameter. Moreover, we assume that $a(\phi)=\phi$, this being the case of the most common distributions in the exponential family, including the Gaussian, gamma, binomial and Poisson distributions. We denote by $V(\cdot)$ the  function satisfying $\text{var}(Y)=V(\mu)\phi$.

In our case, the canonical parameter $\theta$ is a function of both $\boldsymbol{\beta} \in \mathbb{R}^q$   and  $f\in\mathcal{F}.$ We consider the more general penalized log-likelihood
\begin{equation}
\label{equ::pen_like_gen}
\mathcal{L}_{p}(\boldsymbol{\beta}, f)=\mathcal{L}({\boldsymbol{\beta}}, {f})-\frac{\lambda}{2}\ m(f,f),
\end{equation}
where $\mathcal{L}$ is the log-likelihood  and $m(\cdot,\cdot): \mathcal{F} \times \mathcal{F}\rightarrow \mathbb{R}$ is any bilinear, symmetric and semi-positive definite form. This allows us to  introduce  the functional version of the PIRLS algorithm for any functional roughness penalty of this general quadratic form.

We here give a sketch of the derivation of the algorithm and refer to  Appendix \ref{app::PIRLS} for all technical details. For simplicity of writing, we introduce a matrix notation: $\mathbf{y}=(y_1,\dots,y_n)^t$ is the vector of observed data values,  $\boldsymbol{\mu}=(\mu_1,\dots,\mu_n)^t$ is the mean vector,  $\mathbf{X} \in \mathbb{R}^{n\times q}$ denotes the design matrix, whose $i$-th row is given by the covariates  $\mathbf{x}_i$ associated to $y_i$, $\mathbf{f}_n=(f(\mathbf{p}_1),\dots,f(\mathbf{p}_n))^t$ is the vector of evaluations of the spatial field $f$ at the $n$ spatial locations and finally $\mathbf{V}$ is a $n\times n$ diagonal matrix with entries $V(\mu_1),\ldots,V(\mu_n)$, where $V(\cdot)$ is the variance function.

First, we show that the problem of maximizing (\ref{equ::pen_like_gen}) with respect to $(\boldsymbol{\beta},f)$ is equivalent to minimizing the following functional $\mathcal{J}_\lambda\left(\boldsymbol{\beta}, f\right)$ with respect to $(\boldsymbol{\beta},f)$:
$$\mathcal{J}_\lambda\left(\boldsymbol{\beta}, f\right)= \|\mathbf{V}^{-1/2}\left(\mathbf{y}-\boldsymbol{\mu}(\boldsymbol{\beta},f)\right) \|^2 + \lambda\ m(f,f),$$
where $\mathbf{V}$ is considered as fixed, and $\boldsymbol{\mu}(\boldsymbol{\beta}, f)$ is given by the equation (\ref{equ::model_intro}). Since $\mathbf{V}$ in reality depends on $\boldsymbol{\beta}$ and $f$,  this suggests an iterative scheme for the solution of the estimation problem. Let $\boldsymbol{\mu}^{(k)}$ be an estimate of $\boldsymbol{\mu}(\boldsymbol{\beta}, {f})$ after $k$ iterations of the algorithm, and let us consider a first order development of $\boldsymbol{\mu}(\boldsymbol{\beta},f)$ in the neighborhood of the current value $\boldsymbol{\mu}^{(k)}=\boldsymbol{\mu}(\boldsymbol{\beta}^{(k)},f^{(k)}).$ We need to introduce the following notation: $\mathbf{z}^{(k)}$ is the current pseudo-data, defined by $\mathbf{z}^{(k)}=\mathbf{G}^{(k)}(\mathbf{y}-\boldsymbol{\mu}^{(k)})+\boldsymbol{\theta}^{(k)}$, where $ \boldsymbol{\theta}^{(k)}$ is the vector with entries $g(\mu^{(k)}_1),\dots,g(\mu^{(k)}_n)$ and $\mathbf{G}^{(k)}$ is the $n \times n$ diagonal matrix with entries $g'({\mu}^{(k)}_1),\dots,g'({\mu}^{(k)}_n)$; moreover, $\mathbf{V}^{(k)}$ is the current value  of $\mathbf{V}$ for $\boldsymbol{\mu}=\boldsymbol{\mu}^{(k)}$ and $\mathbf{W}^{(k)}= (\mathbf{G}^{(k)})^{-2}(\mathbf{V}^{(k)})^{-1}$.
The first order development of $\boldsymbol{\mu}(\boldsymbol{\beta},f)$ in the neighborhood of the current value $\boldsymbol{\mu}^{(k)}$ is to be considered in the space $\mathbb{R}^q\times \mathcal{F}$ and  yields the following quadratic approximation of $\mathcal{J}_\lambda\left(\boldsymbol{\beta}, f\right)$:
\begin{equation}
\tilde{\mathcal{J}}^{(k)}_\lambda\left(\boldsymbol{\beta}, f\right) =  \|(\mathbf{W}^{(k)})^{1/2}(\mathbf{z}^{(k)}-\mathbf{X}\boldsymbol{\beta}-\mathbf{f}_n)\|^2+ \lambda \ m(f,f),
\label{equ::current_min_prob}
\end{equation}

We may thus consider the following iterative scheme. Let $\boldsymbol{\mu}^{(k)}$ be the value of $\boldsymbol{\mu}$ after $k$ iterations of the algorithm. At the $k+1$ iteration, the following steps are performed:
\begin{enumerate}
\item compute $\mathbf{z}^{(k)}$ and $\mathbf{W}^{(k)}$;
\item \label{item:min} find $\boldsymbol{\beta}^{(k+1)}$ and $f^{(k+1)}$ that jointly minimize (\ref{equ::current_min_prob});
\item set $\boldsymbol{\mu}^{(k+1)}= g^{-1}(\mathbf{X}\boldsymbol{\beta}^{(k+1)} + \mathbf{f}_n^{(k+1)})$.
\end{enumerate}
The stopping criterion is based on a sufficiently small variation of two successive values of the functional (\ref{equ::current_min_prob}). The starting value $\boldsymbol{\mu}^0$ is set to $\mathbf{y}$. In the case of binary outcomes, $\boldsymbol{\mu}^0$ is set to $\boldsymbol{\mu}^0= \frac{1}{2}(\mathbf{y}+\frac{1}{2})$.

When a canonical parameter is used, the log-likelihood of an exponential family distribution is strictly concave. Since the penalization term is concave too, the maximum of the penalized log-likelihood is unique, when it exists. Therefore, if the convergence of the functional PIRLS algorithm is reached, it always results in the maximum penalized log-likelihood estimate. In the simulations and application shown in this paper, just a very few iterations (less than 10) of the algorithm were sufficient to reach convergence, as it is usually the case for generalized linear models.

Step \ref{item:min} of the algorithm still involves a minimization problem over an infinite dimensional space.  In  the case where the penalty has the form $m(f,f)=\int_{\Omega} \big(\Delta f(\mathbf{p})\big)^2d\mathbf{p}$, this minimization problem can be solved extending the methodology described in \cite{Sangalli13}. This extension will be the object of the next section. However, the functional version of the PIRLS algorithm  applies more generally to any type of quadratic roughness penalty.

\section{Penalized least-square problem and finite elements}
\label{sec::fem}
We now focus on  the case where the roughness penalty has the form $m(f,f)=\int_{\Omega} \big(\Delta f(\mathbf{p})\big)^2d\mathbf{p}.$ At each iteration of the functional PIRLS algorithm, we thus have to find the values of ${\boldsymbol{\beta}}\in \mathbb{R}^q $ and ${f}\in\mathcal{F}$ that jointly minimize
\begin{equation}
\tilde{\mathcal{J}}_\lambda\left(\boldsymbol{\beta}, f\right) =  \|(\mathbf{W}^{1/2}(\mathbf{z}-\mathbf{X}\boldsymbol{\beta}-\mathbf{f}_n)\|^2+ \lambda \ \int_{\Omega}(\Delta f)^2.
\label{eq:estimation_problem}
\end{equation}
To simplify the notation we drop here and in the following  the dependence on $k,$ the iteration counter.
 Let us then consider what kind of space $\mathcal{F}$ is well-suited for the problem here considered. To do this, we need to introduce the Sobolev space  $H^m(\Omega)$: this is the Hilbert space of all functions which belong to $L^2(\Omega)$ along with all their distributional derivatives up to the order $m$. Since the roughness penalty term $\int_\Omega (\Delta f)^2$ must be well defined, we need $\mathcal{F}\subset H^2(\Omega)$. Note that by the Sobolev embedding theorem,  $H^2(\Omega)\subset C^0(\Omega)$. Thus, a function $f\in H^2(\Omega)$ is  continuous and  can hence be evaluated at pointwise locations, so that it is possible to compute the vector  $\mathbf{f}_n$ in the least-square term (or in the log-likelihood). Moreover, to ensure  uniqueness of the minimizer of (\ref{eq:estimation_problem}), suitable boundary conditions are required.  Boundary conditions are a way to impose a desired behaviour to the estimated function $f$ at the boundaries of the domain of interest. Typically, we can impose conditions on the value of $f$ at the boundary $\partial \Omega$, that is $\left.f\right|_{\partial \Omega}=\gamma_D$ (\emph{Dirichlet} type boundary conditions), or on the flux of the function through the boundary, that is $\left.\partial_\mathbf{n} f \right|_{\partial \Omega} = (\nabla f)^t \mathbf{n}=\gamma_N$  (\emph{Neumann} type boundary conditions), where $\mathbf{n}$ denotes the outward-pointing normal unit vector to the boundary and $\nabla f = (\partial f/\partial p_1, \partial f/ \partial p_2)^t$ is the gradient of the function $f$.  When the functions $\gamma_D$ or $\gamma_N$ coincide with null functions, the condition is said homogeneous. Moreover, it is possible to  impose different boundary conditions on different portions of the boundary, forming a partition of $\partial \Omega$. To ensure the uniqueness of the minimization problem (\ref{eq:estimation_problem}), we here consider the space:
\begin{equation}
\nonumber
\mathcal{F}=H^2_{\mathbf{n}_0}=\left\{f\in H^2 \ | \ (\nabla f)^t\mathbf{n}=0 \text{ on } \partial\Omega \right\}.
\end{equation}
The interested reader is referred to \cite{Azzimonti1,Azzimonti2} for the case of general boundary conditions.

\subsection{Characterization of the  solution to the penalized least-square problem}
In the following, we assume that the design matrix $\mathbf{X}$ has full rank and that the weight matrix $\mathbf{W}$ has strictly positive entries. Let $ \mathbf{H}= \mathbf{X}(\mathbf{X}^t\mathbf{W}\mathbf{X})^{-1}\mathbf{X}^t \mathbf{W}$, and $\mathbf{Q} = \mathbf{I} - \mathbf{H}$, where $\mathbf{I}$ is an identity matrix of appropriate dimension. Moreover, for any function $u$ in the considered functional space $\mathcal{F}=H^2_{\mathbf{n}_0}(\Omega),$ we denote by $\mathbf{u}_n= u(\mathbf{p}_1),\dots,u(\mathbf{p}_n)$ the vector of evaluations of  $u$ at the $n$ spatial locations. Finally, we denote by $\tilde{\boldsymbol{\beta}}$ and $\tilde{f}$ the minimizers of the penalized least-square functional $\tilde{\mathcal{J}}_\lambda\left(\boldsymbol{\beta}, f\right)$ in (\ref{eq:estimation_problem}), and  by $\hat{\boldsymbol{\beta}}$ and $\hat{f}$ the maximizers of the penalized log-likelihood functional $\mathcal{L}_{p}(\boldsymbol{\beta}, f)$ in (\ref{equ::pen_like}). Under these assumptions, the following Proposition characterizes the minimizers $\tilde{\boldsymbol{\beta}}$ and $\tilde{f}$ of the penalized least-square functional (\ref{eq:estimation_problem}).
\begin{proposition}
\label{prop::weight_min_prob}
There exists a unique pair  $(\tilde{\boldsymbol{\beta}},\tilde{f}) \in\mathbb{R}^q \times H^2_{\mathbf{n}_0} $ which minimizes (\ref{eq:estimation_problem}). Moreover,
\begin{itemize}
\item $\tilde{\boldsymbol{\beta}}=(\mathbf{X}^t\mathbf{W}\mathbf{X})^{-1}\mathbf{X}^t\mathbf{W}(\mathbf{z}-\tilde{\mathbf{f}}_n),$ where $\tilde{\mathbf{f}}_n=(\tilde{f}(\mathbf{p}_1),\dots,\tilde{f}(\mathbf{p}_n))^t$,
\item $\tilde{f}$ satisfies:
\begin{equation}
\label{equ::weak_sol}
\mathbf{u}_n^t\mathbf{Q}\ \tilde{\mathbf{f}}_n + \lambda \int_{\Omega} (\Delta u) (\Delta \tilde{f}) = \mathbf{u}_n^t\mathbf{Q}\ \mathbf{z}, \qquad \forall u\in H^2_{\mathbf{n}_0}.
\end{equation}
\end{itemize}
\end{proposition}
\begin{proof}
See appendix \ref{app::proof_1}.
\end{proof}
Using Proposition \ref{prop::weight_min_prob} and the functional version of the PIRLS algorithm presented in Section \ref{sec::just_PIRLS_variational}, we have a characterization of the maximum penalized log-likelihood in the functional space $H^2_{\mathbf{n}_0}(\Omega)$.

\subsection{Solution to the penalized least-square problem}
\label{sec::finit_el_sol}

 In this section, we describe the methodology yielding to the solution of the problem of minimizing $\tilde{\mathcal{J}}_{\lambda}(\boldsymbol{\beta}, f)$ with respect to both $\boldsymbol{\beta}$ and $ f$. As stated by the proposition \ref{prop::weight_min_prob}, given $\tilde{f}$, it is easy to compute $\tilde{\boldsymbol{\beta}}$. Then, the crucial point is to find $\tilde{f}$ that satisfies (\ref{equ::weak_sol}). For this purpose, we introduce the space:
$$H_{\mathbf{n}_0}^1= \left\{f \in H^1\ |\ (\nabla f)^t \mathbf{n}=0 \text{ on }\partial\Omega \right\}.$$
 Then, as shown in \cite{Sangalli13},   problem (\ref{equ::weak_sol}) is equivalent to finding  $(\tilde{f},\tilde{h})\in H_{\mathbf{n}_0}^1(\Omega)\times H_{\mathbf{n}_0}^1(\Omega)\,$ such that
\begin{equation}
\label{equ::variational_problem}
\left\{
\renewcommand{\arraystretch}{2}
\begin{array}{ll}
\displaystyle \mathbf{u}_n^t \mathbf{Q} \ \tilde{\mathbf{f}}_n - \lambda \int_{\Omega} (\nabla u)^t \nabla \tilde{h} =\mathbf{u}_n^t \mathbf{Q} \ \mathbf{z} \\
\displaystyle -\int_{\Omega} (\nabla \tilde{f})^t \nabla v = \int_{\Omega} \tilde{h}\ v.
\end{array}
\right.
\end{equation}
for any $(u,v) \in H^1_{\mathbf{n}_0} \times H_{\mathbf{n}_0}^1(\Omega)$.  This formulation requires less regularity on the functions involved with respect to  formulation (\ref{equ::weak_sol}), defined in $H^2_{\mathbf{n}_0}(\Omega).$ In the following section, we show how we can use the finite element method  to construct a finite dimensional subspace of $H_{\mathbf{n}_0}^1(\Omega)$, and hence to compute an approximate solution to (\ref{equ::variational_problem}) in such space.

\subsection{Finite elements}
\label{sec::finit_el}
 The finite element method is widely used in engineering applications to numerically solve problems involving partial differential equations \citep[see, e.g.,][]{Quarteroni14}.

To construct a finite element space, we start by partitioning the domain of interest $\Omega$ into small subdomains.  Convenient domain partitions are given for instance by triangular meshes. Figure \ref{gr:port_data_and_tri}, right panel, shows for example a triangulation of the domain of interest for the study of criminality in the city of Portland.  We consider a regular triangulation $\mathcal{T}$ of  $\Omega$, where adjacent triangles share either a vertex or a  complete edge. The domain $\Omega$ is hence approximated by the domain $\Omega_{\mathcal{T}}$ consisting of the union of all triangles, so that the boundary $\partial \Omega$ of $\Omega$ is approximated by a polygon (or more polygons, in the case for instance of domains with interior holes). 
The triangulation is able to describe accurately the complex domain geometry, with its strong concavities corresponding to the river and detailed local structures such as the bridges that connect the two sides of the city center.

Starting from the triangulation, locally supported polynomial functions are defined over the triangles,  providing a set of basis functions $\psi_1,\dots,\psi_K,$ that span a finite dimensional subspace $\mathcal{F}_K$ of $H_{\mathbf{n}_0}^1$. Linear finite elements are for instance obtained considering a basis system where each basis function $\psi_i$ is associated with a  vertex $\boldsymbol{\xi}_i, i=1,\dots, K,$ of the triangulation $\mathcal{T}.$  This basis function $\psi_i$  is a piecewise linear polynomial which takes the value one at the vertex $\boldsymbol{\xi}_i$ and the value zero on all the other vertices of the mesh, i.e., $\psi_i(\boldsymbol{\xi}_j)=\delta_{ij}$, for all $i,j=1,\dots, K$, where $\delta_{ij}$ denotes the Kronecker symbol. Figure \ref{gr:line_el_fin} shows an example of such linear finite element basis function on a planar mesh, highlighting the locally supported nature of the basis.
\begin{figure}[htbp]
\centering
\includegraphics[scale=0.3]{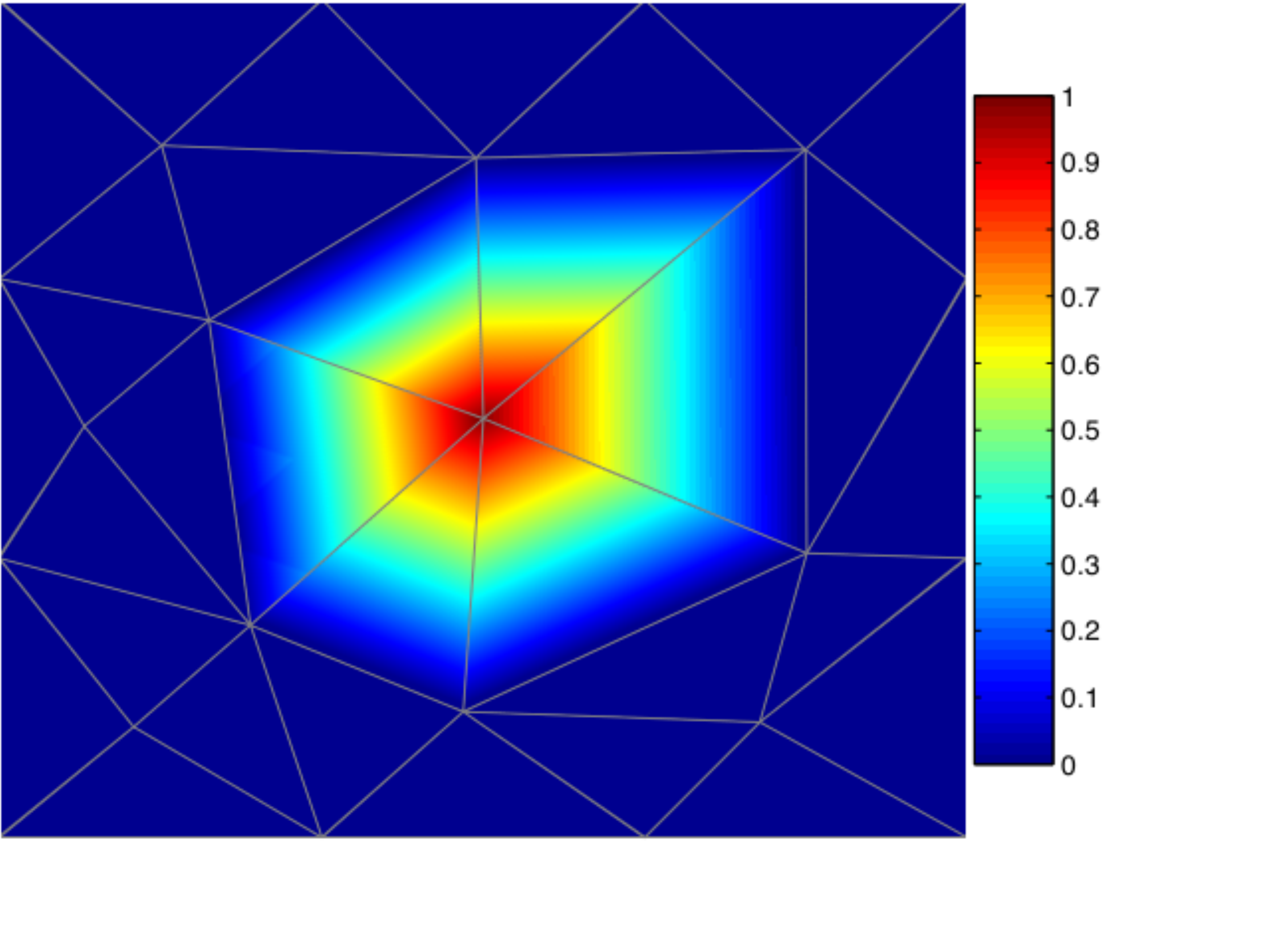}
\caption{\label{gr:line_el_fin}Linear finite element basis function.}
\end{figure}

Now, let $\boldsymbol{\psi}=(\psi_1,\ldots,\psi_K)^t$ be the column vector collecting the $K$ piecewise linear basis functions associated with the $K$ vertices $\boldsymbol{\xi}_1, \dots, \boldsymbol{\xi}_K$. Each function $h$ in the finite element space $\mathcal{F}_K$ can be represented as an expansion in terms of the basis functions $\psi_1,\ldots,\psi_K$. Let $\mathbf{h}=(h_1,\dots,h_K)$ be the coefficients of the basis expansion of $h$, that is the coefficients such that
$$ h(\cdot)=\sum_{j=1}^K h_j \psi_j(\cdot)=\mathbf{h}^t \boldsymbol{\psi}(\cdot).$$
Note that
$$h(\boldsymbol{\xi}_i)=\sum_{j=1}^K h_j \psi_j(\boldsymbol{\xi}_i) = \sum_{j=1}^K h_j \delta_{ij} = h_i, $$
hence
$$\mathbf{h}=\big(h(\boldsymbol{\xi}_1),\ldots,h(\boldsymbol{\xi}_K)\big),$$
which exhibits the fact that each function $h\in \mathcal{F}_K$ is fully characterized by its evaluations on the mesh nodes.
\subsection{Numerical solution to the penalized least-square problem}
The functions and integrals in (\ref{equ::variational_problem}) can be approximated using functions in the finite element space $\mathcal{F}_K,$ so that  problem (\ref{equ::variational_problem}) is approximated with its discrete counterpart: find $(\tilde{f}, \tilde{h})\in \mathcal{F}_K\times \mathcal{F}_K$ that satisfy (\ref{equ::variational_problem}) for any $(u,v)\in \mathcal{F}_K \times \mathcal{F}_K,$ where the integrals are now computed over the triangulation ${\Omega_\mathcal{T}}.$ Let $\boldsymbol{\Psi}$ be  the $n\times K$ matrix of the evaluations of the $K$ basis at the $n$ data locations $\mathbf{p}_1,\dots,\mathbf{p}_n,$
\begin{equation}
\label{eq::Psi}
\boldsymbol{\Psi}=\left[
\begin{array}{c}
\boldsymbol{\psi}^t (\mathbf{p}_1)\\
\vdots \\
\boldsymbol{\psi}^t(\mathbf{p}_n)
\end{array}
\right]
\end{equation}
and  consider the $K\times K$ matrices
$$
\mathbf{R_0}= \int_{\Omega_{\mathcal{T}}}(\boldsymbol{\psi}\ {\boldsymbol{\psi}}^t) \qquad  \quad  \mathbf{R_1}= \int_{\Omega_{\mathcal{T}}} (\nabla\boldsymbol{\psi})^t \nabla \boldsymbol{\psi}.
$$
Using this notation, for functions $\tilde{f},\tilde{h},u,v\in \mathcal{F}_K,$ we can write the integrals in (\ref{equ::variational_problem}) as follows:
$$
\int_{\Omega_\mathcal{T}}  (\nabla u)^t \nabla \tilde{h} =\mathbf{u}^t \mathbf{R}_1 \tilde{\mathbf{h}} ,\hspace{0.75cm}
\int_{\Omega_\mathcal{T}}  (\nabla \tilde{f})^t \nabla v
=\tilde{\mathbf{f}}^t\mathbf{R}_1 \mathbf{v},\hspace{0.75cm}
\int_{\Omega_\mathcal{T}} \tilde{h}\ v=\tilde{\mathbf{h}}^t\mathbf{R}_0 \mathbf{v},
$$
where $\tilde{\mathbf{f}}, \tilde{\mathbf{h}},\mathbf{u}$ and $\mathbf{v}$ are the vectors of the basis expansions of the functions  $\tilde{f},\tilde{h},u$ and $v$ respectively. The discrete counterpart of the problem  (\ref{equ::variational_problem}) thus reduces to solving a linear system, as stated in the following proposition.
\begin{proposition}
\label{prop:LinearSystem}
The discrete counterpart of (\ref{equ::variational_problem}) is given by the system
\begin{equation}
\label{equ::big_system_unconst_mesh}
\left[
\begin{array}{cc}
-\boldsymbol{\Psi}^t\ \mathbf{Q}\ \boldsymbol{\Psi} & \lambda \mathbf{R}_1 \\
\lambda \mathbf{R}_1 & \lambda \mathbf{R}_0
\end{array}
\right]
\left[
\begin{array}{c}
\tilde{\mathbf{f}}\\
\tilde{\mathbf{h}}
\end{array}
\right]
=
\left[
\begin{array}{c}
-\boldsymbol{\Psi}^t \mathbf{Q}\ \mathbf{z}\\
\mathbf{0}
\end{array}
\right],
\end{equation}
which admits a unique pair of solutions $\tilde{\mathbf{f}}, \tilde{\mathbf{h}}$ that are respectively the coefficients of the basis expansion of $\tilde{f}$ and $\tilde{h}$.
 \end{proposition}
 \begin{proof}
Uniqueness of the  solution to (\ref{equ::big_system_unconst_mesh}) is ensured by the positive definiteness of the matrices $\mathbf{R}_0$  and $\left(\boldsymbol{\Psi}^t\mathbf{Q}\boldsymbol{\Psi}+\lambda \mathbf{R}_1\mathbf{R}_0^{-1}\mathbf{R}_1\right).$
 \end{proof}

Let  $\boldsymbol{P}= \mathbf{R}_1 \mathbf{R}_0^{-1} \mathbf{R}_1$ and $ \mathbf{S}=\boldsymbol{\Psi}\left(\boldsymbol{\Psi}^t\ \mathbf{Q}\ \boldsymbol{\Psi} +\lambda   \boldsymbol{P}  \right)^{-1} \boldsymbol{\Psi}^t \mathbf{Q}.$
Then, using the functional version of the PIRLS algorithm, and thanks to Propositions \ref{prop::weight_min_prob} and \ref{prop:LinearSystem}, we obtain  the following expressions for the maximizers $\hat{\boldsymbol{\beta}}$  and $\hat{f}$ of the penalized log-likelihood (\ref{equ::pen_like}):
$$\hat{\boldsymbol{\beta}}=(\mathbf{X}^t\mathbf{W}\mathbf{X})^{-1}\mathbf{X}^t\mathbf{W}(\mathbf{I}-\mathbf{S})\mathbf{z},$$
$\hat{f}$ is identified by the vector
\begin{equation}\label{eq:hat_f}
\hat{\mathbf{f}}= \left(\boldsymbol{\Psi}^t\ \mathbf{Q}\ \boldsymbol{\Psi} +\lambda \boldsymbol{P}\right)^{-1} \boldsymbol{\Psi}^t \mathbf{Q}\mathbf{z},
\end{equation}
and the vector of evaluations of  $\hat{f}$ at the $n$ data locations is given by
$$ \hat{\mathbf{f}}_n=\boldsymbol{\Psi}\hat{\mathbf{f}} = \boldsymbol{\Psi}\left(\boldsymbol{\Psi}^t\ \mathbf{Q}\ \boldsymbol{\Psi} +\lambda \boldsymbol{P}\right)^{-1} \boldsymbol{\Psi}^t \mathbf{Q}\mathbf{z}= \mathbf{S}\mathbf{z},$$
where the vector of pseudo-data $\mathbf{z}$, and the matrices $ \mathbf{W}$, $\mathbf{Q}$ and  $\mathbf{S}$   are those obtained at the convergence of the  PIRLS algorithm.

The positive definite
 matrix $\boldsymbol{P}$ represents the discretization of the penalty
term in (\ref{equ::pen_like}) and (\ref{eq:estimation_problem}). Notice that, thanks to the variational formulation (\ref{equ::variational_problem})
of the estimation problem, this penalty matrix does not involve
the computation of second-order derivatives. \cite{Azzimonti2} shows that, in the finite
element space used to discretize the problem, $\boldsymbol{P}$  is in fact
equivalent to the penalty matrix that would be obtained as
direct discretization of the penalty term in (\ref{equ::pen_like}) and (\ref{eq:estimation_problem}), and involving the
computation of second-order derivatives.


We can define the hat (or influence) matrix $\mathbf{M}$ for the generalized additive model  \citep[see][p.\ 156]{Hastie90} as the matrix satisfying:
\begin{equation}
\nonumber
\hat{\boldsymbol{\boldsymbol{\theta}}}= \mathbf{M}\mathbf{z},
\end{equation}
where $\hat{\boldsymbol{\boldsymbol{\theta}}}$ and $\mathbf{z}$ are respectively the canonical parameter and the pseudo-data at the convergence of the  PIRLS algorithm. In this case, the hat matrix is given by:
$$\mathbf{M} =(\mathbf{H}+\mathbf{Q}\mathbf{S}).$$
The trace of the influence matrix can be used as measure of the equivalent degrees of freedom of the model \citep{Buja89}.
Finally, the fitted mean is given by:
$$\hat{\boldsymbol{\mu}}=g^{-1}(\hat{\boldsymbol{\theta}})= g^{-1}(\mathbf{M}\mathbf{z}).$$

\subsection{Estimation of the scale parameter and selection of the smoothing parameter}
Any distribution of the exponential family is described by two parameters, the mean $\mu$ and the scale parameter $\phi$. The estimation of the mean does not require the estimation of the scale parameter but only of the canonical parameter. To estimate the scale parameter, we must estimate the mean for all the observations. A classical estimator of the scale parameter is \citep[see, e.g.,][]{Wood06}:
\begin{equation}
\label{equ::phi_est}
\hat{\phi}=\frac{\|\mathbf{V}^{-1/2}(\mathbf{y}-\hat{\boldsymbol{\mu})}\|^2}{n-\tr(\mathbf{M})},
\end{equation}
where $\hat{\boldsymbol{\mu}}$ is the estimated mean at the convergence,  $\mathbf{V}$ is the $n \times n$ diagonal matrix with entries $V(\hat{\mu}_1),\ldots,V(\hat{\mu}_n)$ and $\mathbf{M}$ is the hat matrix.
We may choose the smoothing parameter $\lambda$  by minimizing the Generalized Cross Validation (GCV)  criterion \citep{Craven78}:
\begin{equation}
\label{equ::def_GCV}
\text{GCV}(\lambda)= \frac{n \|\mathbf{y} - \boldsymbol{\mu}(\hat{\boldsymbol{\beta}}, \hat{\mathbf{f}})(\lambda)\| ^2}{[n-\gamma\ \tr [\mathbf{M}(\lambda) ] ]^2},
\end{equation}
where $\boldsymbol{\mu}(\hat{\boldsymbol{\beta}}, \hat{\mathbf{f}})(\lambda)$ is the fitted mean at the convergence of the algorithm, for a fixed $\lambda$, and $\gamma$ is a constant factor usually set equal  to $1$. In some cases, the GCV optimum leads to overfitting, so it can be useful to give more weight to the equivalent degrees of  freedom of the model setting $\gamma \geq 1$.

As discussed extensively in \cite{Wood06}, two alternative schemes can be adopted for the selection of the smoothing parameter when using a  PIRLS algorithm. The parameter estimation can be done as a step of the PIRLS algorithm, leading to an update of the value of $\lambda$ at each iteration of the algorithm; alternatively, the update of the smoothing parameter can be done at the convergence of the algorithm. These two different approaches are refereed to as \emph{performance iteration} and \emph{outer iteration} respectively. In this work we shall use an outer iteration scheme.


\section{Model version for areal data}
\label{sec::areal_data_model}

The proposed model can also be specified for the case of areal observations. Specifically, let   $D_1, \dots, D_n$ be disjoints subregions of  the domain $\Omega$. Over  each  subdomain $D_i,$ we observe the realization $y_i$ of a real variable of interest $Y_i$ and a vector of covariate information $\mathbf{x}_i\in\mathbb{R}^q$. We assume $Y_1,\dots,Y_n$ are independent, with $Y_i$  having a distribution within the exponential family, with mean $\mu_i$ and common scale paramenter $\phi$. We  now model $\mu_i$ by
$$g(\mu_i)= \mathbf{x}_i^t \boldsymbol{\beta} + \int_{D_i} f,$$
where the integral of the spatial field $f$ over the subdomain $D_i$  replace the pointwise evaluation of the field considered in model (\ref{equ::model_intro}). We estimate  $\boldsymbol{\beta} \in \mathbb{R}^q$   and the spatial field $f\in\mathcal{F}$ by maximizing a penalized log-likelihood functional in (\ref{equ::pen_like}).
If we  redefine $\mathbf{f}_n$ as $\mathbf{f}_n = ( \int_{D_1} f , \dots, \int_{D_n} f)^t,$ i.e., as being the vector of integrals of the spatial field  over the subdomains, and we redefine the $n\times K $ matrix $\boldsymbol{\Psi}$ in (\ref{eq::Psi}) as the matrix with entry $(i,j)$ given by $\int_{D_i}\psi_j,$ then the derivation of the functional PIRLS algorithm, the implementation of the model and its properties follows as described in the previous sections for the geostatistical data case.

Appendix \ref{app::linear_operator} presents in fact a more general formulation of the model proposed in this work, that comprehends as special cases the model version for geostatistical data and the one for areal data. The results detailed in the previous sections  and  in Appendices \ref{app::spatial-variation}, \ref{app::PIRLS} and \ref{app::proof_1}, for the case of geostatistical data, carry over to this more general model, and thus also to the areal data case.

\section{Simulation studies}
\label{sec::sim}

\subsection{Geostatistical data}
\label{sec::sim-geostat}

In order to illustrate the good performances of the proposed model, we show some simulations on a horseshoe domain \citep{Ramsay02, Wood08} and using the spatial test field shown in the top left panel of Figure \ref{gr:sim_summary_1}, that is detailed in Appendix \ref{app::test-fields}. We consider an outcome with a gamma distribution; in this case, we need to estimate both the  canonical and the scale parameter. We generate $n=200$ data locations uniformly on the horseshoe. We then consider these locations as fixed. For each sampled data location $\mathbf{p}_i$, we generate two independent covariates $x_{1i}$ and $x_{2i}$ having a translated beta distribution; specifically we set $x_{1i}=1+u_{1i}$ and $x_{2i}=1+u_{2i},$ where $u_{1i}\stackrel{iid}{\sim} \text{Beta}(1.5,2) $ and $u_{2i}\stackrel{iid}{\sim} \text{Beta}(3,2)$. We set $\beta_1=-\frac{2}{5}$ and $\beta_2=\frac{3}{10}$.
For each sampled data location $\mathbf{p}_i$, we then generate independent gamma random variables, with mean $\mu_i=-(\mathbf{x}_i^t \boldsymbol{\beta} + f(\mathbf{p}_i))^{-1}$ and common scale parameter $\phi$. We repeat this simulation $M=100$ times.

The top right panel of Figure \ref{gr:sim_summary_1} shows the sampled data in a simulation repetition, with the size of the point marker proportional to data values. The bottom center and right panels of the same Figure displays the  scatter plots  of the response  versus the two covariates; from these plots is not apparent that the two covariates are significant in explaining the response.

\begin{figure}[tbhp]
\centering
\includegraphics[scale=0.5]{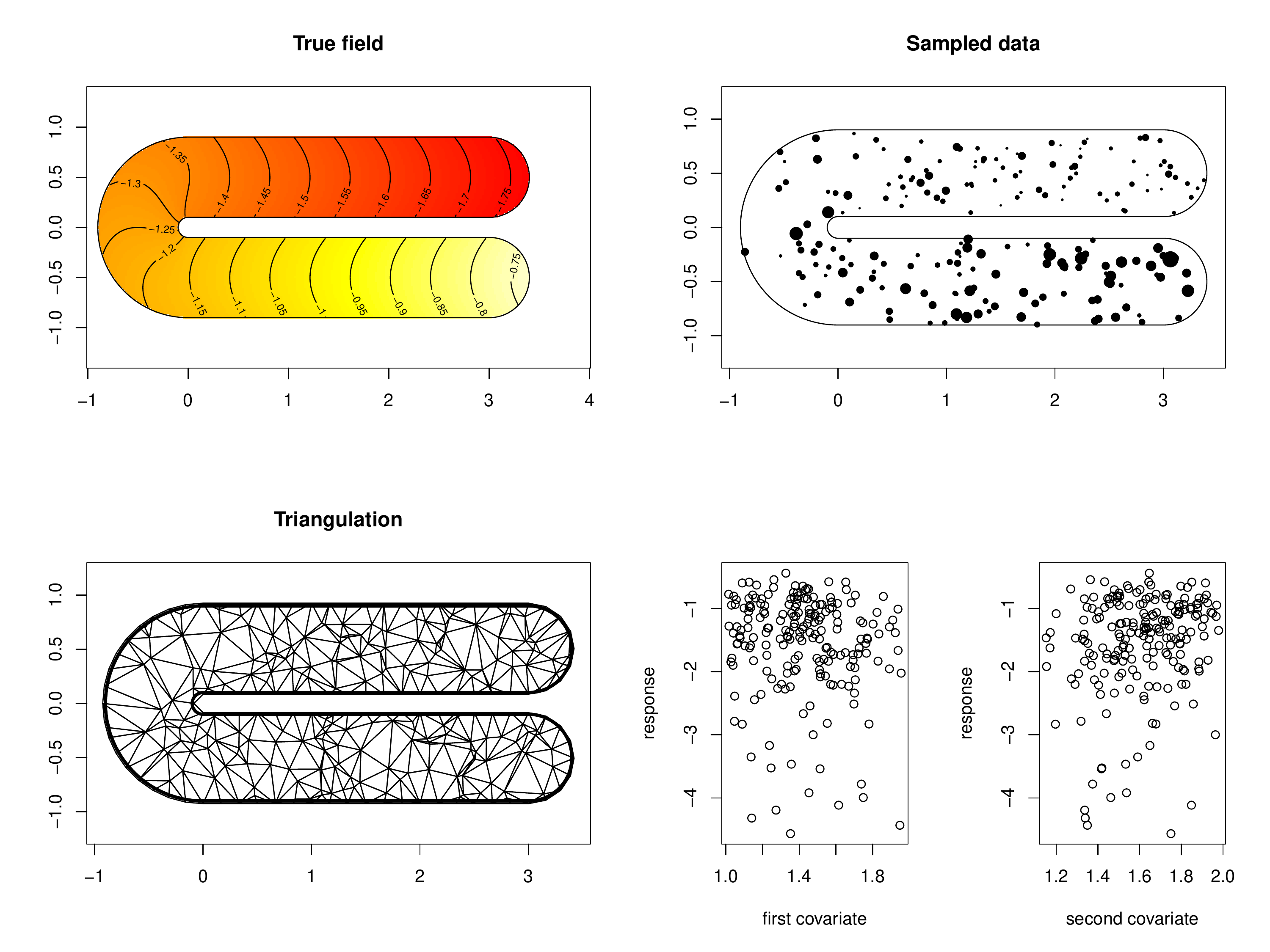}
\caption{\label{gr:sim_summary_1}Simulation with geostatistical data. Top left: the true field $f$ to be estimated. Top right: the data sampled in the first simulation repetition; the marker size is proportional to data values; the data locations are considered as fixed. Bottom left: the triangulation used to obtain the GSR-PDE estimate; this is a contrained Delaunay triangulation of the locations of the data shown in the top right panel. Bottom center and bottom right: scatter plots of the response versus the two covariates, for the first simulation repetition.}
\end{figure}

\begin{figure}[tbhp]
\centering
\includegraphics[scale=0.5]{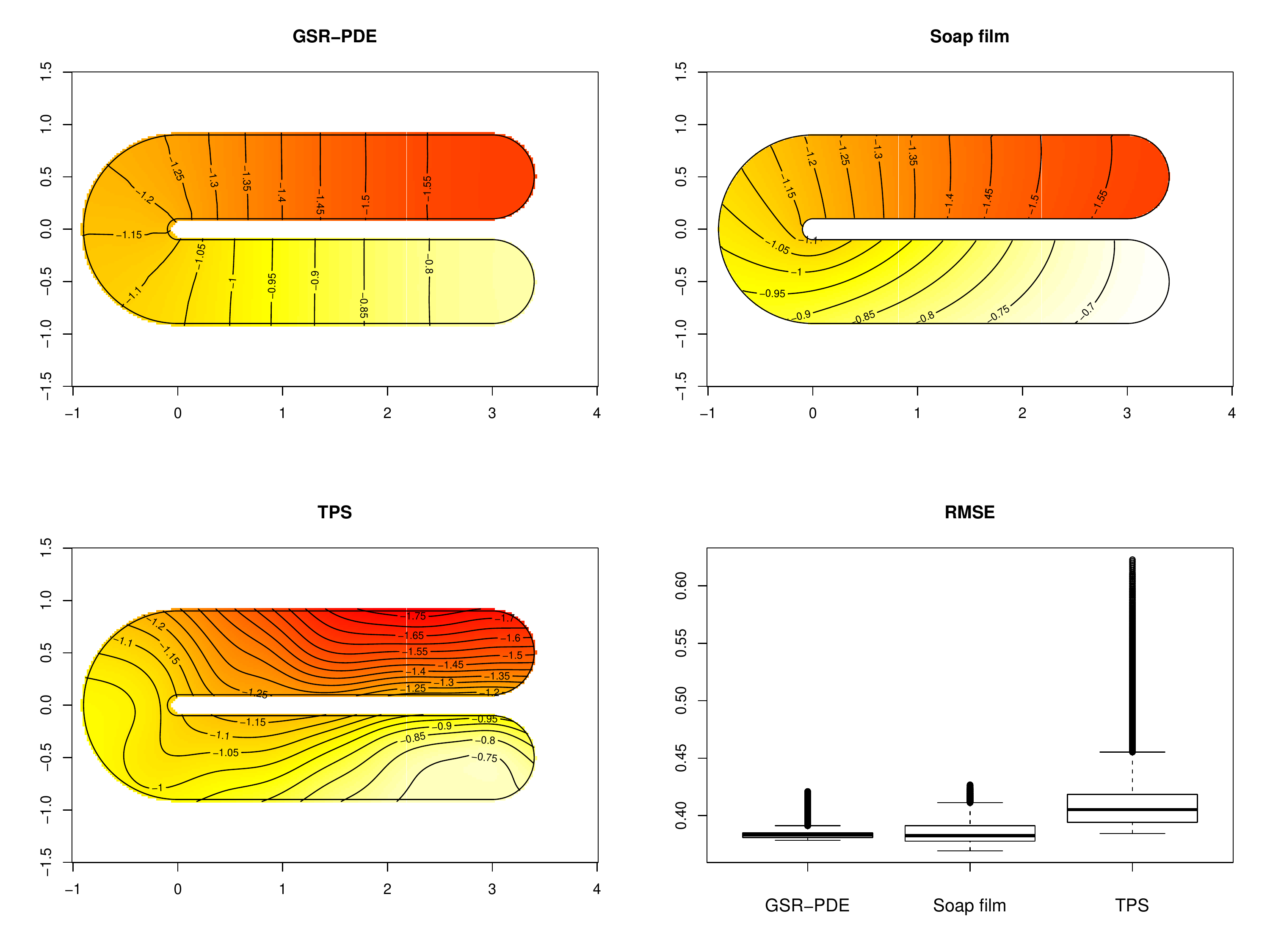}
\caption{\label{gr:sim_summary_2}Simulation with  geostatistical data. Estimates of the spatial field obtained in the first simulation repetition by GSR-PDE (top left), Soap film  (top right) and TPS (bottom left). On the bottom right, the boxplot of the spatial distributions of the RMSE of the three spatial field estimators over the $M=100$ simulation repetitions, computed on a fine grid of points over the horseshoe domain.}
\end{figure}

We compare our method to soap film smoothing \citep{Wood08} and to the thin-plate splines \citep{Duchon77, Wahba90}, implemented using the \textbf{R} package \texttt{mgcv} \citep{mgcv}. Soap film smoothing  uses 72 degrees of freedom, as in the implementation  given in the reference manual of the \texttt{mgcv} package (see function \texttt{Predict.matrix.soap.film}). Thin-plate splines (TPS)  uses the default settings with $40$ degrees of freedom. For  the proposed GSR-PDE method, we use linear finite elements with a triangular mesh that is a constrained Delaunay triangulation of the $n$ data locations; see Figure \ref{gr:sim_summary_1}, bottom left panel. To ensure that the  comparison is fair, at each simulation repetition we select the smoothing parameter for each of the considered methods  optimizing the GCV criterion in an outer iteration scheme.

The root mean squared error (RMSE),  over the $M=100$ simulation repetitions, of the estimators of ${\boldsymbol{\beta}}$ are comparable accross the three considered methods (the RMSE of $\hat{\beta}_1$ are:  0.151 for GRS-PDE,  0.150 for Soap and 0.159 for TPS; the RMSE of $\hat{\beta}_2$ are:  0.178 for GRS-PDE,  0.174 for Soap and 0.178 for TPS). The bottom right panel of Figure \ref{gr:sim_summary_2} shows the  boxplots of the spatial distribution of the RMSE, over the $M=100$ simulation repetitions, of the estimators of the spatial field $\hat{f};$ specifically, we consider a fine grid of  points $\mathbf{p}$, of step $0.02$ in the $x$-direction and $0.01$ in the $y$-direction, and for each of these points $\mathbf{p}$ we
 compute
$$\text{RMSE}(\mathbf{p})= \sqrt{\frac{1}{M}\sum_{j=1}^M \left(\hat{f}(\mathbf{p})- f(\mathbf{p})\right)^2}.$$
These boxplots show that the proposed GRS-PDE method and Soap film smoothing provide significantly better estimates than thin-plate splines. The reason of this comparative advantage is  highlighted  by the spatial field estimates returned by the three methods in the first simulation replicate, shown in the first three panels of  Figure \ref{gr:sim_summary_2}. The thin-plate spline technique is blind to the shape of the domain and smooths across the internal boundaries: the higher values of the field in one side of the horseshoe domain are smoothed with the lower values of the field in the other side of the domain, returning an highly biased estimate. The proposed GRS-PDE method and soap film smoothing do not suffer this problem,  accurately complying with the domain geometry. The proposed GRS-PDE method is the best technique in terms of RMSE of the spatial field estimator.


\subsection{Areal data}
\label{sec::sim-areal}

\begin{figure}[tbhp]
\centering
\includegraphics[width=1\textwidth]{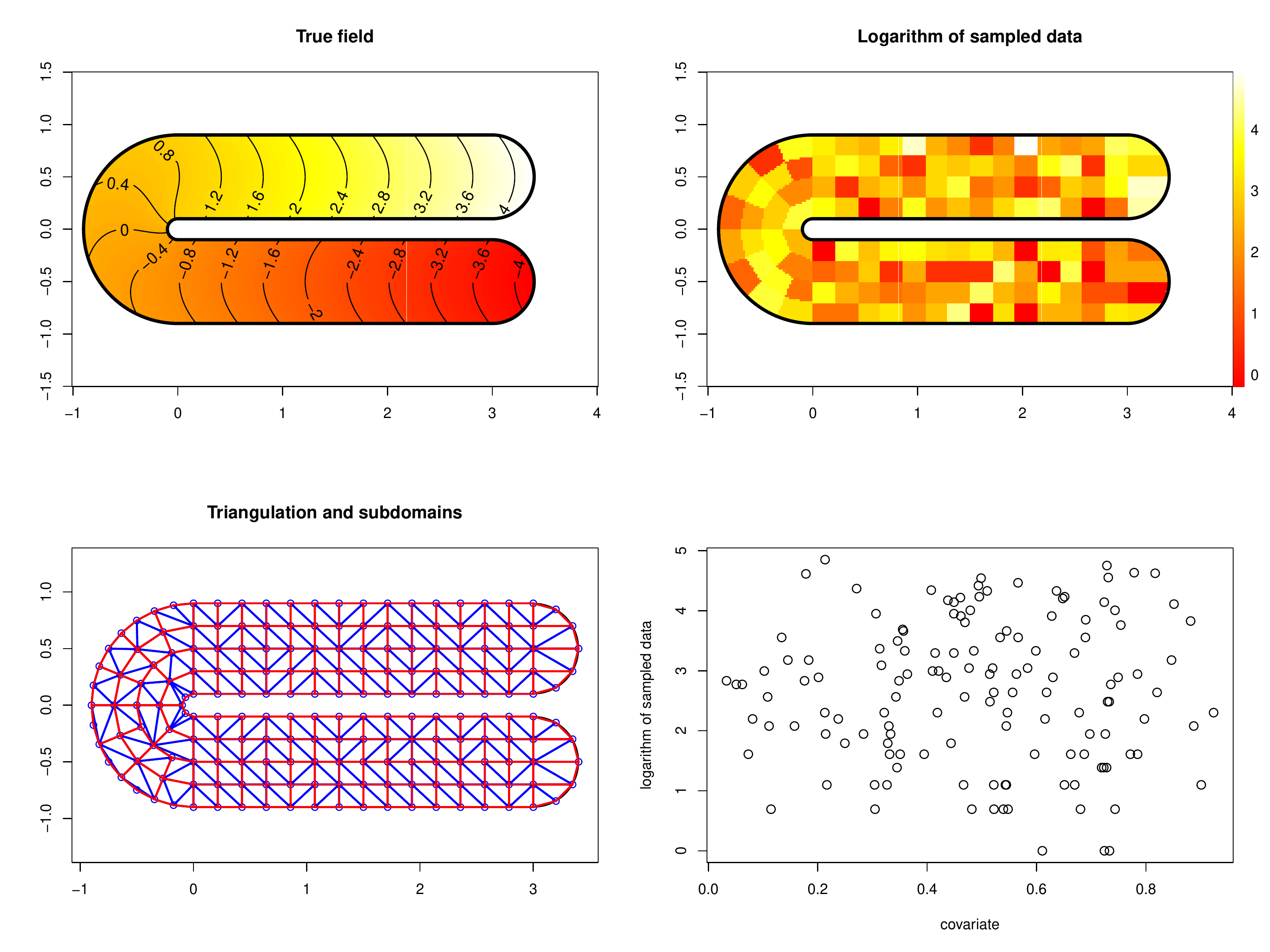}
\caption{\label{gr::summary_areal_1}Simulation with areal data. Top left: the true field to be estimated. Top right: the logarithm of sampled data. Bottom left: the triangulation used to obtain the GSR-PDE estimate, with the borders of the subdomains highlighted in red. Bottom right: scatter plots of the response versus the covariate.}
\end{figure}

We now present a simulation with areal data.
We consider the test function of the horseshoe domain displayed in the left panel of Figure \ref{gr::summary_areal_1} and detailed in Appendix \ref{app::test-fields}. The bottom left panel of the same figure shows in red the borders of the $ n= 142$ sub-domains  $D_i$ considered.  Indipendently over each subdomain $D_i,$ we generate a covariate $x_{i}$ having beta distribution: $x_{i}\stackrel{iid}{\sim} \text{Beta}(2,2)$.
We set $\beta=5$. Over each subdomain $D_i,$ we the generate independent  Poisson random variables with mean $\mu_i$, where  $\log(\mu_i) = \mathbf{x}_i^t \boldsymbol{\beta} +\int_{D_i}f.$ Notice that in this case the scale parameter is 1 and does not need to be estimated.
The simulation is repeated $M=100$ times.

The top right panel of Figure \ref{gr::summary_areal_1} shows the sampled data in the first simulation repetition (in logarithmic scale); the bottom right panel of the same Figure displays a scatter plot  of the response (in logarithmic scale) versus the covariate.


%

 \begin{figure}[tbhp]
\centering
\includegraphics[width=0.4\textwidth]{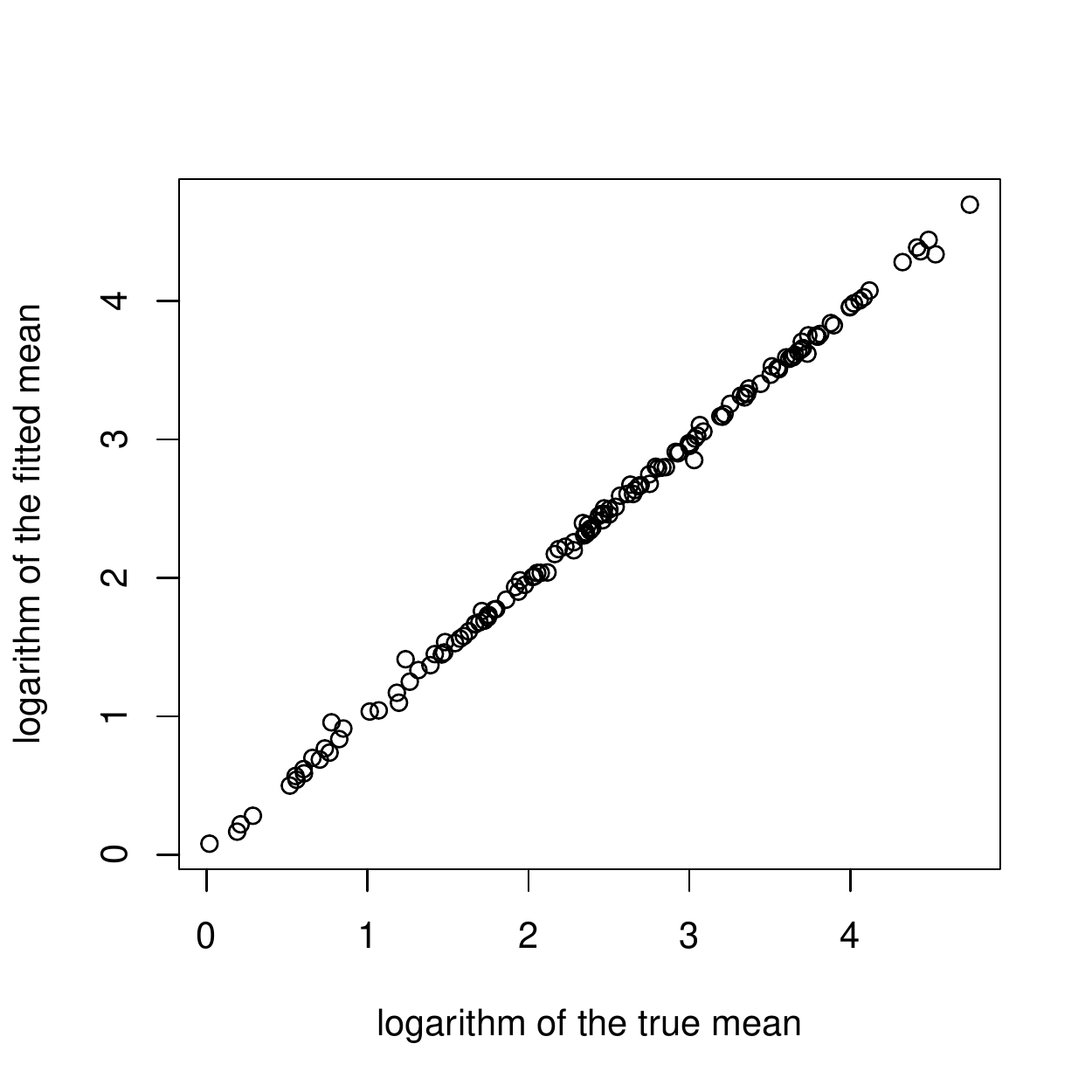}
\includegraphics[width=0.4\textwidth]{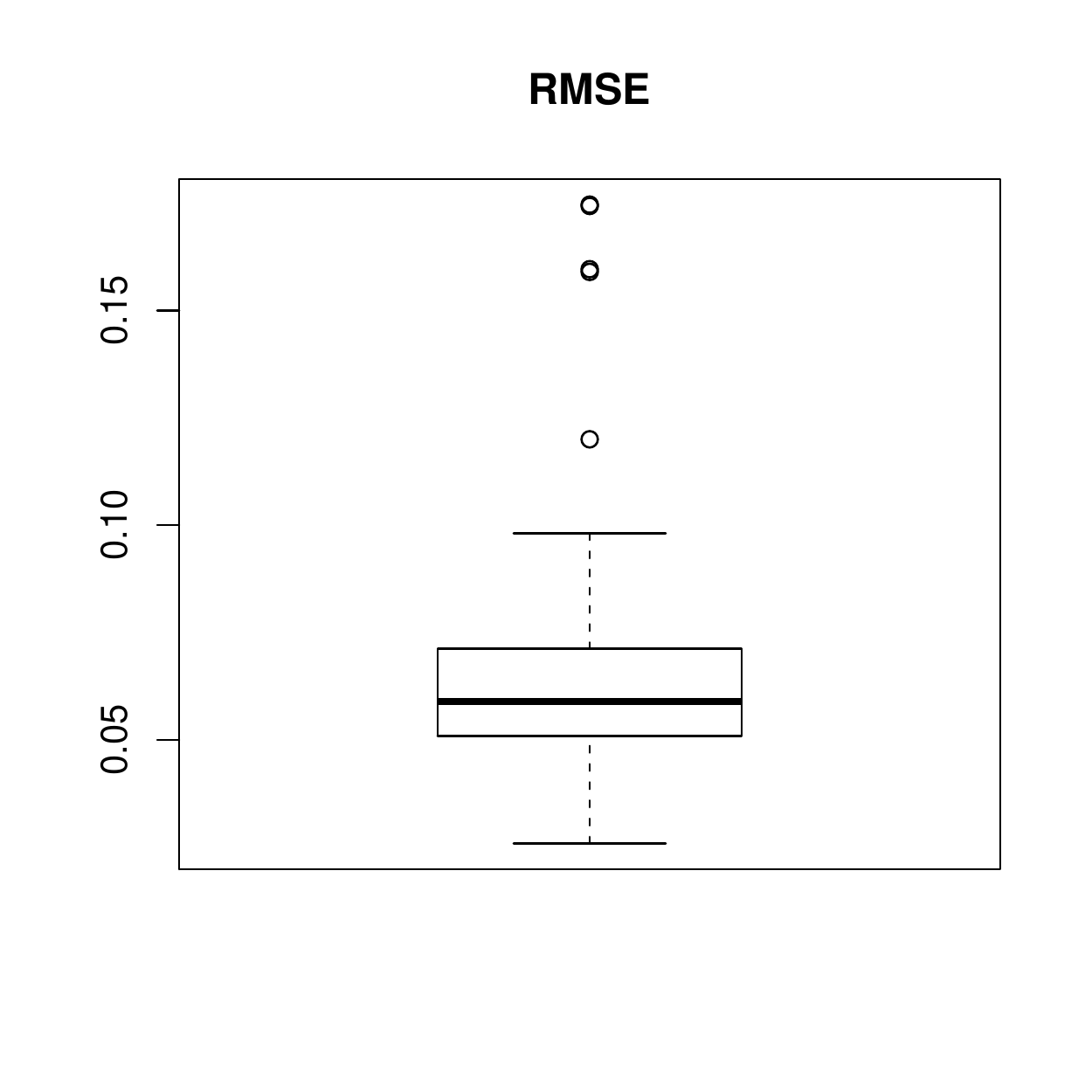}
\caption{\label{gr::scatter}Simulation with areal data. Left: scatter plot of the the estimated mean  vs  the true mean over each subdomain, for  the first simulation repetition. Right: boxplot of the spatial distribution of the RMSE, over the $M=100$ simulation repetitions, of the non-parametric part of the model, computed over each subdomain as detailed in eq.\  (\ref{eq:RMSE-areal}).}
\end{figure}

\begin{figure}[tbhp]
\centering
\includegraphics[width=1\textwidth]{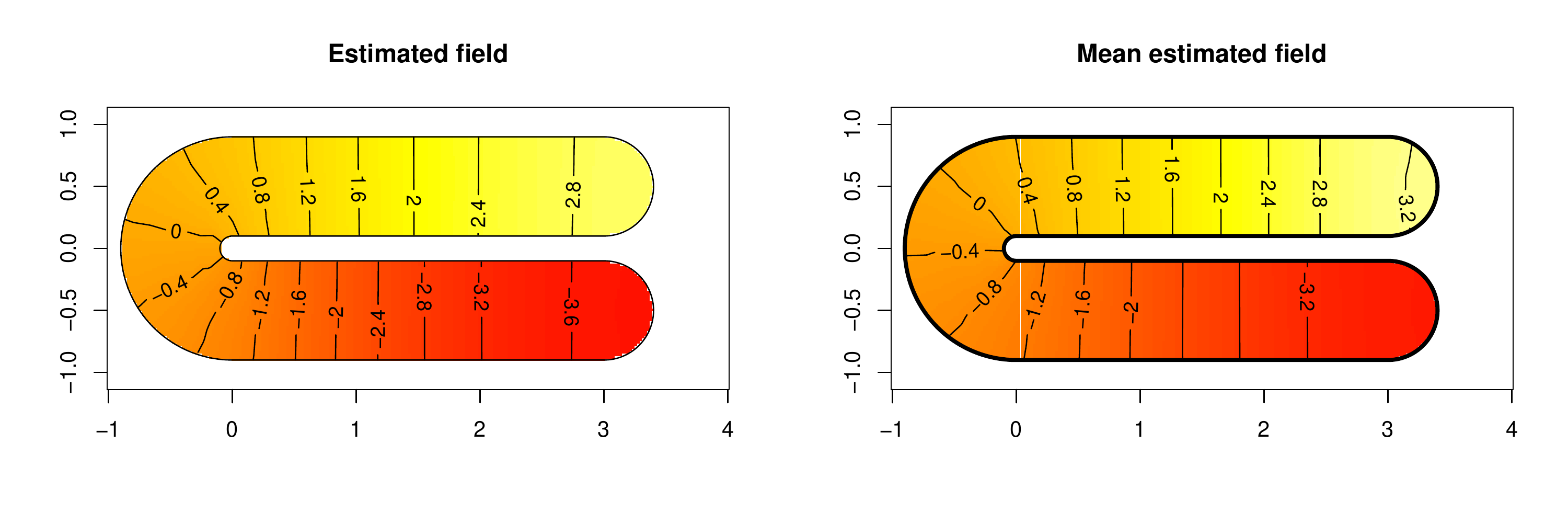}
\caption{\label{gr::summary_areal_2}Simulation with areal data: Left: estimated spatial field in the first simulation repetition. Right: mean of the estimated spatial fields over the $M=100$ simulation repetitions.}
\end{figure}

The sample mean of the estimated $\hat{\beta}$ coefficient over the $M=100$ simulation repetitions is  5.021 (true value: 5) with a standard deviation of  0.073 and a RMSE of 0.076.
 The left panel of Figure \ref{gr::scatter} compares the log  estimated mean over each subdomain in the first simulation repetition, and the true one, showing the very good performances of the method. The right panel of the same figure displays the boxplot of the spatial distribution of the RMSE, over the $M=100$ simulation repetitions, of the non-parametric part of the model, computed over each sub-domain as
\begin{equation}\label{eq:RMSE-areal}
  \text{RMSE}(D_i)= \sqrt{\frac{1}{M}\sum_{j=1}^M \left(\int_{D_i}\hat{f}- \int_{D_i}f\right)^2}
\end{equation}
 for $i=1,\ldots, n=142.$
 Finally, Figure \ref{gr::summary_areal_2} shows the estimated field $\hat{f}$ in the first simulation repetition and the sample mean of the estimated spatial field over the 100 simulation repetitions. These highlight that the method is able to recover  quite well the pointwise values of the field, even though using only areal observations. The estimates of  $\beta$  and of $f$ appear to have a negligible bias and a small variance.


\section{Application: Crimes in Portland}
\label{sec::application}

The city of Portland, Oregon (USA) has made publicly available a data set about all crimes committed in the city in 2012\footnote{Portland crime data: \url{http://www.civicapps.org/datasets}}. We would like  to study the  criminality over this city, taking into account auxiliary census information\footnote{Census Bureau Data:  \url{http://www.portlandoregon.gov/oni/28387}}. The census information (year 2010) is aggregated at the level of the neighborhoods. Here in particular we consider as covariate the total population of each neighborhood. The map of Portland in  the left panel of Figure \ref{gr:port_data_and_tri} highlights in blue the borders of these neighborhoods, together  with the locations of crimes. For computational simplicity, the triangulation of the city territory shown in the right panel of the same figure has been constructed in a way to comply with the borders of the neighborhoods. Since the covariate is only available at the level of neighborhoods, we decided to aggregate also the crimes, thus considering as response variable the total crime count over each neighborhood.
We model these data as an inhomogeneous Poisson process. Specifically,  the total crime counts $Y_1,\ldots,Y_n$ over the $n=98$ neighborhoods are modeled as independent Poisson random variables with mean  $\mu_i$, where
$$\log(\mu_i)=\log(\texttt{pop}_i)\beta + \int_{D_i} f\ dx,$$
and $\texttt{pop}_i$ denotes the total population over the i-th neighborhood.

We select the smoothing parameter via the GCV criterion. We get $\hat{\beta}=  0.381$, confirming that the population density contributes positively to the crimes count in a given neighborhood. Figure \ref{gr::crime_fit_neigh} compares, in a logarithm scale, the  observed and estimated crime densities over each neighborhood, where the density is computed as the total crime count over the neighborhood divided by the area of the neighborhood. These are also compared in the scatter plot in the left panel of Figure \ref{gr::crime_dist_Port}, highlighting the goodness of fit of the model. Finally, the right panel of Figure \ref{gr::crime_dist_Port} shows the estimated spatial field. When the estimated field is close to zero, the crimes count is well described by the parametric part of the model, namely as a rate of the number of residents.
 The highest levels of the estimated spatial field are located dowtown; this is likely due to the high number of people who come to the city center for work or  leisure during the day and in the evening. The estimate complies with the complex shape of the domain.

\begin{figure}[hbtp]
\centering
\includegraphics[scale=.54]{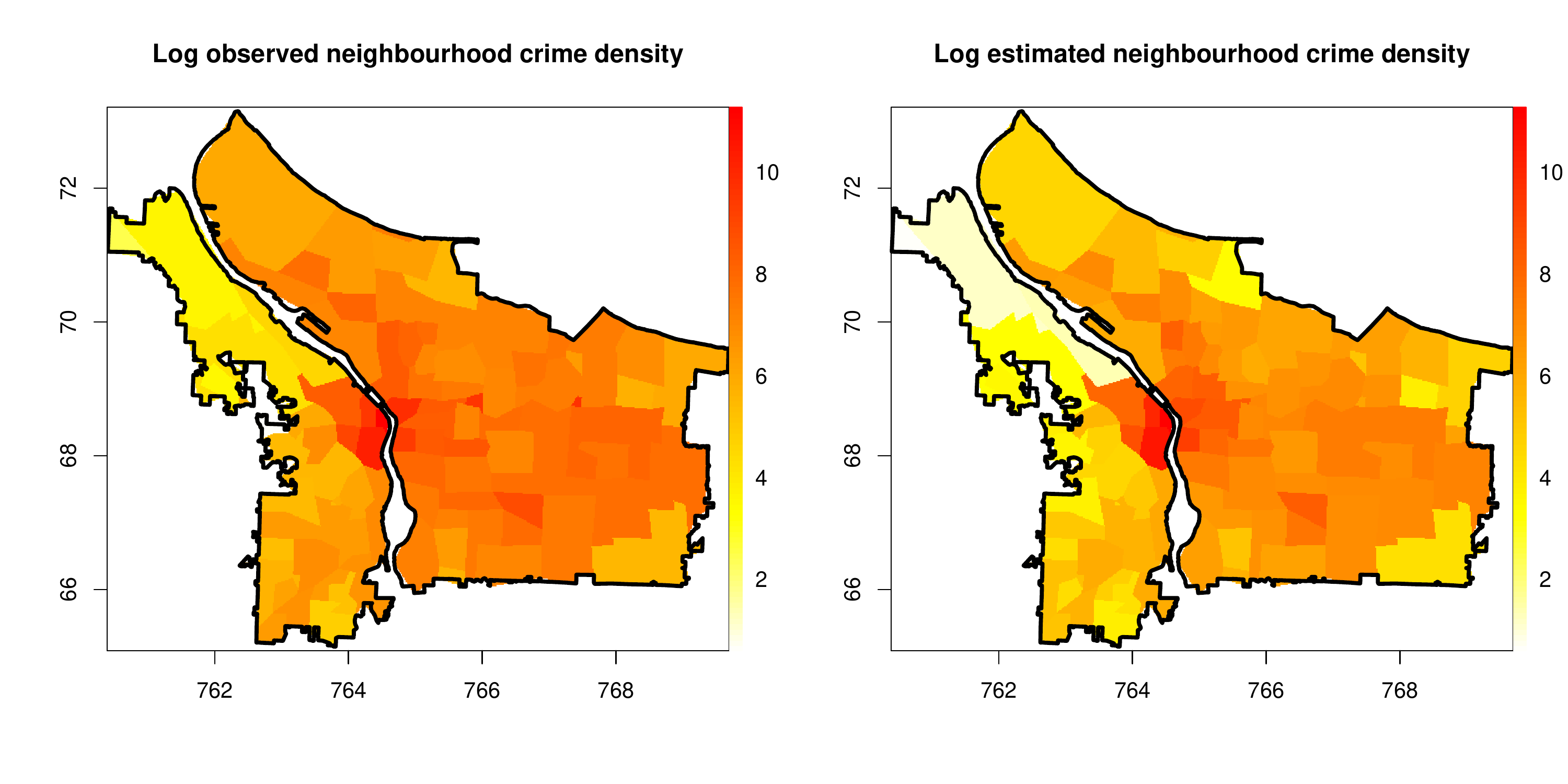}
\caption{\label{gr::crime_fit_neigh} Log of observed crime density over each neighborhood (left) and corresponding log  of estimated crime density  (right).}
\end{figure}

\begin{figure}[htbp]
\begin{subfigure}{.29\linewidth}
\includegraphics[scale=.63]{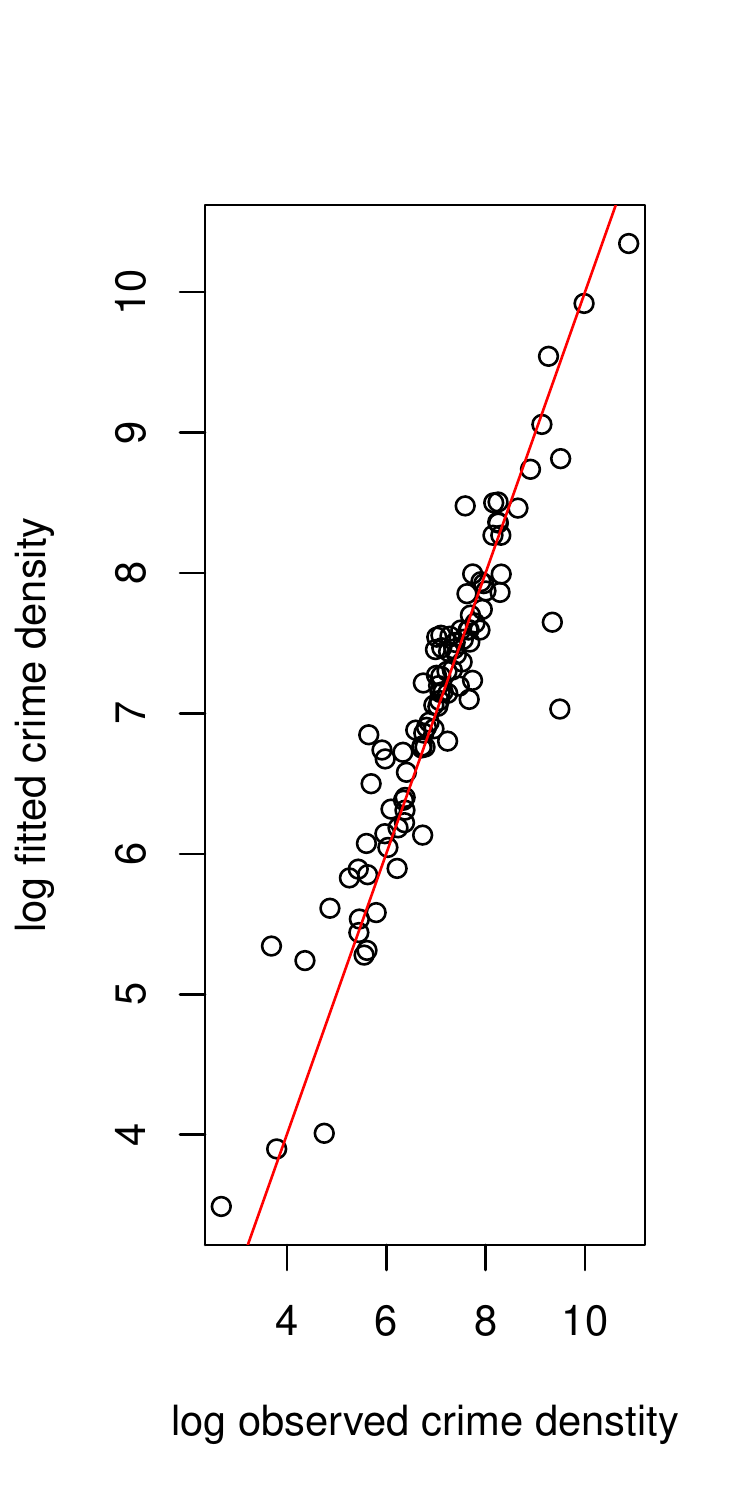}
\end{subfigure}
\begin{subfigure}{.68\linewidth}
\includegraphics[scale=.22]{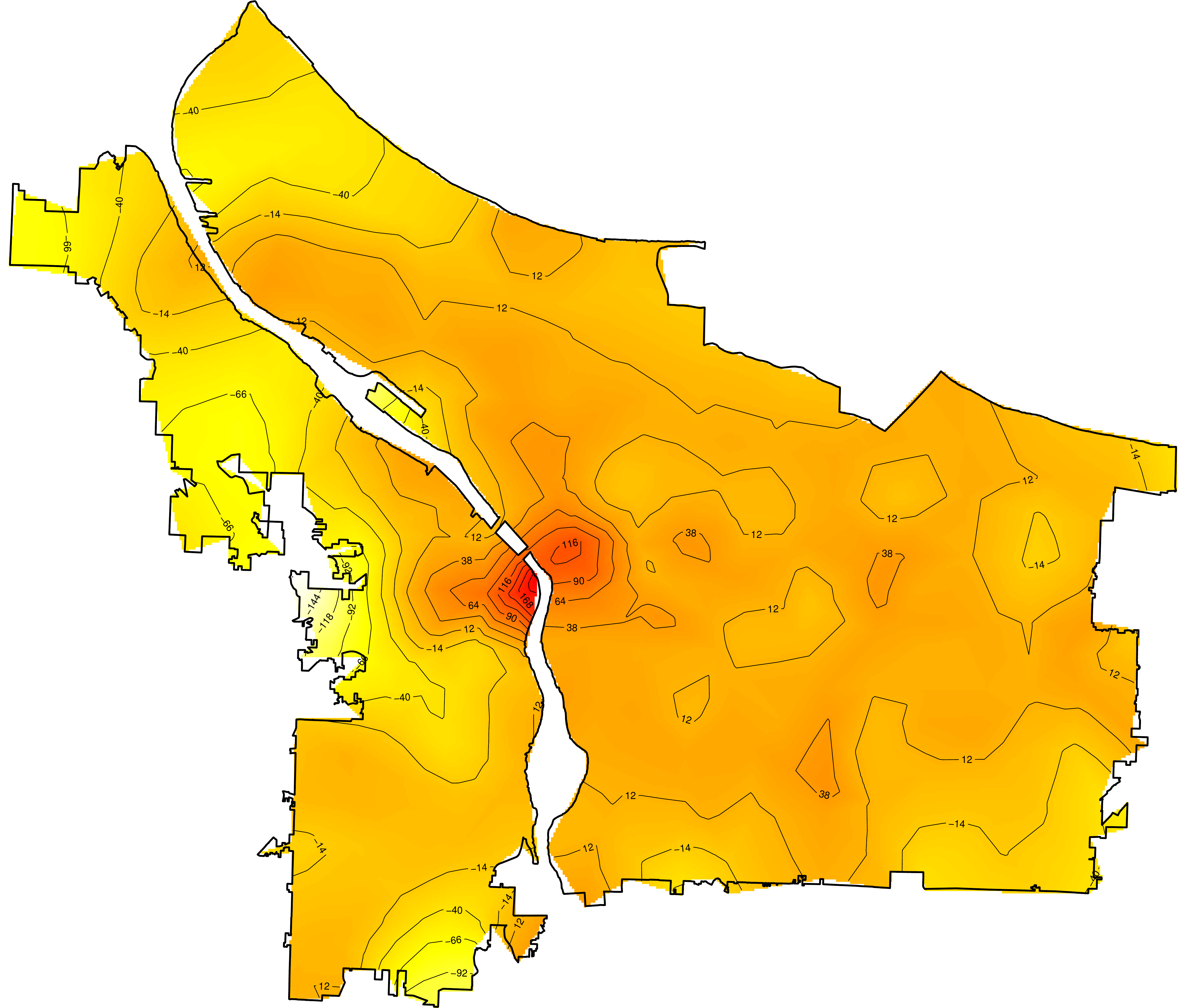}
\end{subfigure}
\caption{\label{gr::crime_dist_Port}Left: fitted crime density per neighborhoods vs observed one, in a logarithmic scale. Right: estimated spatial field $f$ over the city of Portland.}
\end{figure}


\section{Discussions}
\label{sec:discussions}

The method proposed can be extended in various directions.  First of all, owing  to the functional version of the PIRLS algorithm, our methodology can be extended to more complex  roughness penalties. This is particularly interesting when a priori knowledge is available on the problem under study, that can be formalized in terms of a partial differential equation modeling the phenomenon behavior.  \cite{Azzimonti2} shows for instance  that in some applications using a penalty based on a priori knowledge about the problem can dramatically improve the accuracy of the estimation. By using more complex roughness penalties we can also  account for spatial anisotropy and non-stationarity. See also Appendix \ref{app::spatial-variation}. Moreover, as mentioned in Section \ref{sec::fem}, it is possible to use different kinds of boundary conditions, allowing for a very flexible modelling of the behavior of the spatial field at the boundary of the domain of interest.

Furthermore, following the approach developed in \cite{Ettinger12}, the proposed method could be extended  to deal with data distributed over curved domains, specifically over surface domains. This would permit to tackle important applications in the geosciences, dealing for instance with Poisson counts and other type of variables of interest observed over the globe or over regions with complex orography. Other fascinating fields of applications of this modeling extension would be in the neurosciences and other life sciences, studying for instance signals associated to neuronal activity over the cortical surface, the  highly convoluted thin sheet
of neural tissue that constitutes the outermost part of the brain.

Other numerical techniques and associated basis could also be used to solve the estimation problem, instead of the finite element method here considered. For instance, B-Splines and NURBS \citep{Piegl97} are extensively used in computer-aided design (CAD),
manufacturing, and engineering, to represent the 3D surface of the designed item.  \cite{Wilhelm15} offers a first example of a spatial data analysis model exploiting these basis, thus avoiding the domain approximation implied by finite elements. Extending  this approach to the generalized linear setting here considered would further broaden the applicability of the proposed methodology to many engineering fields, including the automotive, the aircraft and the space sectors.


\section*{Aknowledgements}
This work has been developed while M.\ Wilhelm was visiting the Department of Mathematics of Politecnico di Milano, funded by the European  Erasmus program. M. Wilhelm would like to thank Yves Till\'e for his constant support and encouragement, Bree Ettinger and Laura Azzimonti for comments and helps. The authors are grateful to Victor Panaretos for insightful  comments, to  Timoth\'ee Produit for his help using QGIS and to Lionel Wilhelm for his support to construct  the mesh over the city of Portland. L.M.\ Sangalli acknowledges funding by MIUR Ministero dell'Istruzione dell'Università e della Ri\-cer\-ca, FIRB Futuro in Ricerca Starting Grant project \lq\lq Advanced statistical and numerical methods for the analysis of high dimensional functional data in life sciences and engineering\footnote{\url{http://mox.polimi.it/users/sangalli/firbSNAPLE.html}}.


\bibliography{wil_sang}
 \bibliographystyle{plainnat} 

\addcontentsline{toc}{chapter}{Bibliography}


\appendix

\section{On the spatial variation structure of the estimator}
\label{app::spatial-variation}

Let us focus on the special case of the proposed model, where the outcomes $y_i$ are normally distributed with mean $\mu_i=\theta_i=x_i^t \boldsymbol{\beta} +{f}(\mathbf{p}_i)$ and constant variance $\sigma^2$. In this case, the maximization of the penalized log-likelihood function is
  equivalent to the minimization of the penalized least-square functional considered in \cite{Sangalli13} and the quadratic form of the functional  allows to characterize analytically the solution to the estimation problem, so that there is no need to resort to the functional version of the PIRLS algorithm. Specifically, the estimator $\hat{f}$ of the spatial field is identified by the vector
\begin{equation}\label{eq::hat_f_normal}
  \hat{\mathbf{f}}= \left(\boldsymbol{\Psi}^t \mathbf{Q} \boldsymbol{\Psi} +\lambda \boldsymbol{P}\right)^{-1} \boldsymbol{\Psi}^t \mathbf{y},
\end{equation}
where the vector of observed data values $\mathbf{y}$ has replaced the vector of speudo-data $\mathbf{z}$ that is found in equation (\ref{eq:hat_f});  the matrix $\mathbf{Q} = \mathbf{I} - \mathbf{H}$ represents the contribution of the parametric part of the model, where now $\mathbf{H}= \mathbf{X}(\mathbf{X}^t\mathbf{X})^{-1}\mathbf{X}^t$. In this case, the estimator of the spatial field is linear in the observed data values, and  has a typical penalized regression form,  $\boldsymbol{P}$ being the discretization of the regularizing term.
From (\ref{eq::hat_f_normal}), it follows that  the estimate of the field $f$ at any generic location $\mathbf{p}\in\Omega$ is  given by
\begin{equation} \nonumber \hat f(\mathbf{p})=\boldsymbol{\psi}(\mathbf{p})^t (\boldsymbol{\Psi}^t \boldsymbol{\Psi} + \boldsymbol{P})^{-1}\boldsymbol{\Psi}^t \mathbf{y}. \end{equation}
Its mean and variance are given by
\begin{eqnarray}
\E[\hat f(\mathbf{p})]&=&  \boldsymbol{\psi}(\mathbf{p})^t (\boldsymbol{\Psi}^t \boldsymbol{\Psi} +\boldsymbol{P})^{-1}\boldsymbol{\Psi}^t \hat{\mathbf{f}}_n\nonumber \\
\nonumber \var[\hat f(\mathbf{p})]&=&  \sigma^2 \boldsymbol{\psi}(\mathbf{p})^t (\boldsymbol{\Psi}^t \boldsymbol{\Psi} +\boldsymbol{P})^{-1}\boldsymbol{\Psi}^t \boldsymbol{\Psi}(\boldsymbol{\Psi}^t \boldsymbol{\Psi} + \boldsymbol{P})^{-1}\boldsymbol{\psi}(\mathbf{p})
\nonumber
\end{eqnarray}
and the covariance  at any two locations $\mathbf{p}_1, \mathbf{p}_2\in \Omega$  is given by
\begin{eqnarray}
\nonumber
\cov[\hat f(\mathbf{p}_1,),\hat f(\mathbf{p}_2)] &=& \sigma^2 \boldsymbol{\psi}(\mathbf{p}_1)^t (\boldsymbol{\Psi}^t \boldsymbol{\Psi} +\boldsymbol{P})^{-1}\boldsymbol{\Psi}^t \boldsymbol{\Psi}(\boldsymbol{\Psi}^t \boldsymbol{\Psi} + \boldsymbol{P})^{-1}\boldsymbol{\psi}(\mathbf{p}_2).
\end{eqnarray}
The above expressions highlight that both the first order structure and the second order structure of the field estimator are determined by the penalization term. The regularizing term considered in this work induces an isotropic and stationary space variation. Different  regularizations would imply different mean and covariance structures for the field estimator.
 For instance, \cite{Azzimonti2} consider a regularized spatial regression model and show that by changing the regularizing terms and considering more complex differential operators it is possible to include in the model a priori information about the spatial variation of the phenomenon, and model also space anisotropies and non-stationarities. The estimator $\hat{\boldsymbol{\beta}}$ is also linear in the observed data values and it is straightforward to compute its distributional properties.

Outside of the Gaussian case, it is not possible to characterize analytically the solution of the estimation problem. Nevertheless, at each iteration of the functional PIRLS algorithm, the field estimator has the same form and  properties reported above, with respect to the pseudo-data $\mathbf{z}.$ Hence, also in this more general case, the regularizing term implies the first and second order structure of the estimator. Quantification of uncertainty is in this case possible using the techniques developed for
generalized additive models; see, e.g., \cite{Hastie90} and \cite{Wood06}. Bayesian approaches are also available in this context; see, e.g., \cite{Marra12}.

\section{Proof of the functional justification of PIRLS algorithm}
\label{app::PIRLS}
Using the notation given in Section \ref{sec::just_PIRLS_variational}, we want to maximize the penalized log-likelihood function $\mathcal{L}_p(\boldsymbol{\beta}, f)$ of any exponential family distribution, which is given by:
$$
\mathcal{L}_p(\boldsymbol{\beta}, f)= \mathcal{L}(\boldsymbol{\beta}, f)-\frac{\lambda}{2} m(f,f) =  \sum_{i=1}^n y_i \theta_i(\boldsymbol{\beta},f)- b(\theta_i(\boldsymbol{\beta}, f)) - \frac{\lambda}{2} m(f,f),$$
where $\mathcal{L}$ is the likelihood of an exponential family distribution, $b(\cdot)$ is a function depending on the distribution considered and $\theta_i (\boldsymbol{\beta}, f)= g(\mathbf{x}_i^t\boldsymbol{\beta} + f(\mathbf{p}_i))$ is the canonical parameter.
The maximizers ($\hat{\boldsymbol{\beta}}, \hat{f}$) of the functional (\ref{equ::pen_like}) must satisfy the following system of first order equations:
\begin{equation}
\label{equ::syst_variational}
\left\{
\renewcommand{\arraystretch}{2}
\begin{array}{rlll}
\displaystyle \frac{\partial \mathcal{L}(\hat{\boldsymbol{\beta}}, \hat{f})}{\partial \beta_k}&=&0, & \forall k=1,\dots,q, \\
\displaystyle \lim_{t\rightarrow 0}\frac{1}{t}\left[ \mathcal{L}(\hat{\boldsymbol{\beta}}, \hat{f}+t u)-\mathcal{L}(\hat{\boldsymbol{\beta}}, \hat{f})\right] - \lambda\  m(u, \hat{f})&=&0 & \forall u \in \mathcal{F}. \\
\end{array}
\right.
\end{equation}
This system involves the derivatives with respect to both parameters $\boldsymbol{\beta}$ and $f$. The derivative with respect to $f$ is a G\^ateaux derivative in the direction of $u,$ where $u\in\mathcal{F}$. In particular, $ m(u, \hat{f})$ is the derivative  of the term $ m(\hat{f}, \hat{f})$. We then have to compute the terms involving $\mathcal{L}$ only.
We first compute:
$$
\begin{array}{ll}
\mathcal{L}(\boldsymbol{\beta}, f+tu)-\mathcal{L}(\boldsymbol{\beta}, f)=&\\
=\displaystyle \sum_{i=1}^n \frac{1}{\phi} \left[(y_i \theta_i(\boldsymbol{\beta}, f+tu) - b(\theta_i(\boldsymbol{\beta}, f+tu))-y_i \theta_i(\boldsymbol{\beta},f) + b(\theta_i(\boldsymbol{\beta}, f))) \right]\\
= \displaystyle \sum_{i=1}^n \frac{1}{\phi} (y_i \theta_i(\boldsymbol{\beta}, f+tu)-y_i \theta_i(\boldsymbol{\beta},f)- \left( b(\theta_i(\boldsymbol{\beta}, f+tu))-b(\theta_i(\boldsymbol{\beta}, f)) \right).
\end{array}
$$
Dividing this expression by $t$ and taking the limit as $t$ tends to $0$ gives the G\^ateaux derivative of $\mathcal{L}(\boldsymbol{\beta},f)$ in the direction of $u$, denoted by $\frac{\partial \mathcal{L}(\boldsymbol{\beta},u)}{\partial f}$. We have then:
$$
\renewcommand{\arraystretch}{2}
\begin{array}{ll}
 \displaystyle \dfrac{\partial \mathcal{L}(\boldsymbol{\beta},u)}{\partial f}=\lim_{t\rightarrow 0} \frac{1}{t} \left[\mathcal{L}(\boldsymbol{\beta}, f+tu)-\mathcal{L}(\boldsymbol{\beta}, f)\right]\\
 =\displaystyle \lim_{t\rightarrow 0} \frac{1}{t} \left(\sum_{i=1}^n \frac{1}{\phi} \left[y_i \theta_i(\boldsymbol{\beta}, f+tu)-y_i \theta_i(\boldsymbol{\beta},f)- ( b(\theta_i(\boldsymbol{\beta}, f+tu))-b(\theta_i(\boldsymbol{\beta}, f))\right]\right)\\
=\displaystyle \sum_{i=1}^n \frac{1}{\phi} \left[y_i \frac{\partial\theta_i(\boldsymbol{\beta}, u)}{\partial f} - \frac{\partial b}{\partial \theta}(\theta_i(\boldsymbol{\beta}, f)) \frac{\partial \theta_i(\boldsymbol{\beta}, u)}{\partial f}\right].
\end{array}
$$
We then need to compute $\frac{\partial\theta_i(\boldsymbol{\beta}, u)}{\partial f}$.
We recall that, for a  distribution within the exponential family,  $\mathbb{E}\left[Y_i\right]=\mu_i=\frac{\partial b(\theta_i)}{\partial \theta}$ and $\text{var}(Y_i)=\frac{\partial^2 b(\theta)}{\partial\theta^2} \phi$. We thus have:
$$\frac{\partial \mu}{\partial \theta}= \frac{\partial^2 b}{\partial \theta^2} \Rightarrow \frac{\partial \theta}{\partial \mu}= \frac{1}{\frac{\partial^2 b}{\partial \theta^2}}$$
and hence:
$$\displaystyle \frac{\partial\theta_i(\boldsymbol{\beta}, u)}{\partial f}= \dfrac{1}{\frac{\partial^2 b(\theta_i(\boldsymbol{\beta}, u))}{\partial \theta^2}}\frac{\partial \mu_i(\boldsymbol{\beta}, u)}{\partial f }.$$
We can therefore conclude that:
$$\dfrac{\partial \mathcal{L}(\boldsymbol{\beta}, u)}{\partial f}=0 \Leftrightarrow \sum_{i=1}^n \frac{1}{\phi}\frac{(y_i-\frac{\partial b(\theta_i)}{\partial \theta})}{\frac{\partial^2 b(\theta_i)}{\partial \theta^2}}\frac{\partial \mu_i(\boldsymbol{\beta}, u)}{\partial f }= \sum_{i=1}^n\frac{(y_i-\frac{\partial b(\theta_i)}{\partial \theta})}{\text{var}(Y_i)}\frac{\partial \mu_i(\boldsymbol{\beta}, u)}{\partial f }=0.$$
Since we have $\var{(Y_i)}=V(\mu_i)\phi= \frac{\partial^2 b(\theta)}{\partial \theta^2} \phi$ and $\frac{\partial b(\theta_i)}{\partial \theta}=\mu_i$, we finally obtain the following expression for the derivative of the likelihood with respect to the functional parameter:
\begin{equation}\label{eq::der_f}
\dfrac{\partial \mathcal{L}(\boldsymbol{\beta}, u)}{\partial f}=0 \Leftrightarrow \sum_{i=1}^n\frac{(y_i-\mu_i)}{V(\mu_i)}\frac{\partial \mu_i(\boldsymbol{\beta}, u)}{\partial f }=0.
\end{equation}
We now need to compute the derivative of the likelihood with respect to $\beta_{j}$:
$$
\frac{\partial \mathcal{L}(\boldsymbol{\beta})}{\partial \beta_{j}}= \sum_{i=1}^n \frac{1}{\phi} \left(y_i\frac{\partial \theta_i}{\partial \beta_{j}}-\frac{\partial b(\theta_i)}{\partial \theta}\frac{\partial \theta_i}{\partial \beta_{j}}\right).$$
Since
$$\frac{\partial \theta_i}{\partial \beta_{j}}= \frac{\partial \theta_i}{\partial \mu_i}\frac{\partial \mu_i}{\partial \beta_{j}}= \frac{1}{\frac{\partial^2 b(\theta_i)}{\partial \theta^2}}\frac{\partial \mu_i}{\partial \beta_{j}},$$
and using similar computations, we finally get:
\begin{equation}\label{eq::der_beta}
  \frac{\partial \mathcal{L}(\boldsymbol{\beta})}{\partial \beta_{j}}=0 \Leftrightarrow \sum_{i=1}^n\frac{(y_i-\mu_i)}{V(\mu_i)}\frac{\partial \mu_i}{\partial \beta_{j} }=0.
\end{equation}
Putting (\ref{eq::der_f}) and (\ref{eq::der_beta}) together implies that the solution to (\ref{equ::syst_variational}) is equivalent to finding $\boldsymbol{\mu} = \boldsymbol{\mu}(\boldsymbol{\beta}, f)$ that satisfies
\begin{equation}\label{eq::system_der}
\left\{
\renewcommand{\arraystretch}{2}
\begin{array}{rlll}
\displaystyle \sum_{i=1}^n\frac{(y_i-\mu_i)}{V(\mu_i)}\frac{\partial \mu_i}{\partial \beta_{j} }&=&0, & \forall j=1,\dots,q, \\
\displaystyle \sum_{i=1}^n\frac{(y_i-\mu_i)}{V(\mu_i)}\frac{\partial \mu_i(\boldsymbol{\beta}, u)}{\partial f } + \lambda\  m(f,u) &=&0 & \forall u \in \mathcal{F}. \\
\end{array}
\right.
\end{equation}
If we now assumed that $V(\mu_i)$ is constant, solving (\ref{eq::system_der}) would be  equivalent to finding the  minimizers of the following functional
$$\mathcal{J}_\lambda({\boldsymbol{\beta}},{f})= \|\mathbf{V}^{-1/2}\left(\mathbf{y}-\boldsymbol{\mu}\right) \|^2 + \lambda \ m(f,f),$$
where $\mathbf{V}$ is the $n\times n$ diagonal matrix with entries $V(\mu_1),\ldots,V(\mu_n)$. Since in reality $\mathbf{V}$ depends on $\boldsymbol{\mu}$, this suggests an iterative computation scheme.

Let $\boldsymbol{\mu}^{(k)}$ be an estimate of $\boldsymbol{\mu}(\boldsymbol{\beta}, {f})$ after $k$ iterations of such a scheme. At this point, we consider a first order approximation of $\boldsymbol{\mu}$ in the neighbourhood of the current value $\boldsymbol{\mu}^{(k)} = (\boldsymbol{\beta}^k, f^{(k)})$:
$$\boldsymbol{\mu}(\boldsymbol{\beta}, \mathbf{f})\approx\underbrace{g^{-1}(\mathbf{X}\boldsymbol{\beta}^{(k)}+\mathbf{f}_{n}^{(k)})}_{=\boldsymbol{\mu}^{(k)}} + \frac{\partial \boldsymbol{\mu}(\boldsymbol{\beta}, f)}{\partial \boldsymbol{\beta}} (\boldsymbol{\beta}-\boldsymbol{\beta}^{(k)})+  \frac{\partial \boldsymbol{\mu}(\boldsymbol{\beta}, f-f^{(k)})}{\partial f}.$$
We then have to compute the partial derivatives of $\boldsymbol{\mu}$ with respect to both parameters $\boldsymbol{\beta}$ and $f$. Let us start with the derivative with respect to $\boldsymbol{\beta}$. We have: $g(\mu_i^{(k)})=\mathbf{x}^t_i\boldsymbol{\beta}^{(k)}+f(\mathbf{p}_i)$. Taking the derivative with respect to $\beta_{j}$ on both sides, we get:$$g'(\mu_i^{(k)})\frac{\partial \mu_i^{(k)}}{\partial \beta_{j}}=\frac{\partial}{\partial \beta_{j}}g(\mu_i^{(k)})= \frac{\partial}{\partial \beta_{j}}\left( \mathbf{x}_i^t\boldsymbol{\beta}^{(k)} +f(\mathbf{p}_i)\right) = {x}_{ij},$$
where $ {x}_{ij}$ is the $j$th component of the vector $\mathbf{x}_i$, or equivalently  the $ij$th component of the design matrix $\mathbf{X}$.
Then:
$$\frac{\partial \mu_i^{(k)}}{\partial \beta_{j}} = \frac{{x}_{ij}}{g'(\mu_i^{(k)})}$$
that in matrix form is:
$$\frac{\partial \boldsymbol{\mu}(\boldsymbol{\beta}, f)}{\partial \boldsymbol{\beta}}= (\mathbf{G}^{(k)})^{-1}\mathbf{X},$$
where $\mathbf{G}^{(k)}$ is the $n \times n$ diagonal matrix with entries $g'({\mu}^{(k)}_1),\dots,g'({\mu}^{(k)}_n)$.

Let us now compute  the derivative $\boldsymbol{\mu}(\boldsymbol{\beta}, f)$, in the direction $f$. We first recall that $\boldsymbol{\mu}(\boldsymbol{\beta}, f)=g^{-1}(\boldsymbol{\theta}(\beta, f)$ and $\boldsymbol{\theta}=(\mathbf{X}\boldsymbol{\beta}+\mathbf{f}_n)$. For the $i$th component, we then obtain:
$$
\renewcommand{\arraystretch}{2}
\begin{array}{lll}
&\displaystyle \lim_{t \rightarrow 0} \frac{\mu_i(\boldsymbol{\beta}^{(k)}, f^{(k)}+t (f-f^{(k)}))-\mu_i(\boldsymbol{\beta}^{(k)}, f^{(k)})}{t}\\
&\qquad =\displaystyle \lim_{t \rightarrow 0} \frac{g^{-1}(\theta_i(\boldsymbol{\beta}^{(k)}, f^{(k)}+t (f-f^{(k)})))-g^{-1}(\theta_i(\boldsymbol{\beta}^{(k)}, f^{(k)}))}{\theta_i(\boldsymbol{\beta}^{(k)}, f^{(k)}+t (f-f^{(k)}))-\theta_i(\boldsymbol{\beta}^{(k)}, f^{(k)})}\\
&\qquad \displaystyle \cdot \frac{\theta_i(\boldsymbol{\beta}^{(k)}, f^{(k)}+t (f-f^{(k)}))-\theta_i(\boldsymbol{\beta}^{(k)}, f^{(k)})}{t}\\
&\qquad =\displaystyle (g^{-1})'(\theta_i(\boldsymbol{\beta}^{(k)}, f^{(k)}))\left(f(\mathbf{p}_i)-f^{(k)}(\mathbf{p}_i)\right) =   \frac{1}{g'(\mu_i^{(k)})}\left(f(\mathbf{p}_i)-f^{(k)}(\mathbf{p}_i)\right).
\end{array}
$$
Hence, we finally have the following first order approximation of $\mathcal{J}_\lambda\left(\boldsymbol{\beta}, f\right)$ in the neighbourhood of the current value $\boldsymbol{\mu}^{(k)} = (\boldsymbol{\beta}^{(k)}, f^{(k)})$:
\begin{eqnarray*}
  \tilde{\mathcal{J}}^{(k)}_\lambda\left(\boldsymbol{\beta}, f\right) & = & \|\mathbf{V}^{-1/2}\left[\mathbf{y} - \left(\boldsymbol{\mu}^{(k)} +(\mathbf{G}^{(k)})^{-1}\mathbf{X} (\boldsymbol{\beta} - \boldsymbol{\beta}^{(k)})  +  (\mathbf{G}^{(k)})^{-1} (\mathbf{f}_n - \mathbf{f}_n^{(k)}\right)\right]\|^2 + \lambda \ m(f,f)\\
& = & \|\mathbf{V}^{-1/2}(\mathbf{G}^{(k)})^{-1}\left(\mathbf{G}^{(k)}(\mathbf{y} - \boldsymbol{\mu}^{(k)}) + \mathbf{X} \boldsymbol{\beta}^{(k)} + \mathbf{f}_n^{(k)}  - \mathbf{X} \boldsymbol{\beta} - \mathbf{f}_n
\right) \|^2 + \lambda  \ m(f,f) \\
\end{eqnarray*}
Setting $\mathbf{z}^{(k)}=\mathbf{G}^{(k)}(\mathbf{y}-\boldsymbol{\mu}^{(k)})+ \mathbf{X}\boldsymbol{\beta}^{(k)}+ \mathbf{f}_n^{(k)}$, and denoting by $\mathbf{W}^{(k)}$ the $n\times n$ diagonal matrix with ${i}$th entry $V(\mu_i^{(k)})^{-1}g'(\mu_i^{(k)})^{-2}$, we can rewrite
$$\tilde{\mathcal{J}}^{(k)}_\lambda\left(\boldsymbol{\beta}, f\right)= \|(\mathbf{W}^{(k)})^{1/2}(\mathbf{z}^{(k)}-\mathbf{X}\boldsymbol{\beta}-\mathbf{f}_n)\|^2+ \lambda  \ m(f,f).$$
Since $\mathbf{W}^{(k)}$ is positive definite, $\tilde{\mathcal{J}}_\lambda$ is  a quadratic form whose minimum exists and is unique.


\section{Proof of Proposition \ref{prop::weight_min_prob}}
\label{app::proof_1}
Before giving the proof of Proposition \ref{prop::weight_min_prob}, we recall the Lax-Milgram theorem \cite[see, e.g.,][]{Quarteroni14}:
\begin{theorem}[Lax-Milgram]
Let $\mathcal{F}$ be a Hilbert space, $G(\cdot,\cdot):\mathcal{F}\times\mathcal{F}\rightarrow\mathbb{R}$ a continuous and coercive bilinear form and $F:\mathcal{F}\rightarrow\mathbb{R}$ a linear and continuous functional. Then, there exists a unique solution of the following problem:$$ \text{find }u\in \mathcal{F} \text{ such that }\ G(u,v)=F(v),\ \forall v\in \mathcal{F}.$$
Moreover, if $G(\cdot,\cdot)$ is symmetric, then $u \in \mathcal{F}$ is the unique minimizer in $\mathcal{F}$ of the functional $B:\mathcal{F}\rightarrow \mathbb{R}$, defined as
\begin{equation}
\nonumber
B(u)=G(u,u)-2F(u).
\end{equation}
\end{theorem}
We also need the following result.
\begin{lemma}
\label{lem::coerc}
The bilinear and symmetric form $G:H^2_{\mathbf{n}_0}(\Omega)\times H^2_{\mathbf{n}_0}(\Omega)$ defined as $$G(v,u)=\mathbf{u}_n^t \mathbf{Q} \ \mathbf{v}_n + \lambda\int_{\Omega} (\Delta u)(\Delta v),$$
is continuous and coercive.
\end{lemma}
\begin{proof}
We recall the definition of the norm $\|\cdot\|_{H^2}$ and of the semi-norm $|\cdot|_{H^2}$:
$$\| u \|_{H^2} = \sum_{|\alpha| \leq 2} \|\partial^\alpha u\|_{L^2}, \quad |u|_{H^2} = \|\Delta u \|_{L^2}, \quad \forall u \in H^2.$$
First, note that the semi-norm $|\cdot|_{H^2}$ and the norm $\|\cdot\|_{H^2}$ are equivalent in $H^2_{\mathbf{n}_0}(\Omega)$ \cite{Quarteroni14}, i.e, there exists $C_0>0$ such that $ |u|_{H^2(\Omega)}\geq C_0 \|u\|_{H^2(\Omega)}, \ \forall u\in H^2_{\mathbf{n}_0}(\Omega)$. Then, we have:
$$
\renewcommand{\arraystretch}{2}
\begin{array}{lll}
\displaystyle G(u,u)&=&\underbrace{\mathbf{u}_n^t \mathbf{Q} \ \mathbf{u}_n}_{\geq 0} + \lambda\int_{\Omega} (\Delta u)^2\\
\displaystyle  &\geq& \lambda\int_{\Omega} (\Delta u)^2 = \lambda |u|_{H^2(\Omega)} \\
&\geq& \lambda C_0 \|u\|_{H^2(\Omega)};
\end{array}
$$
hence, $G(\cdot,\cdot)$ is coercive.

We then show the continuity of $G(\cdot,\cdot)$. Since $H^2(\Omega)\subset C^0(\Omega)$ and since the norms $\|\cdot\|_{\infty}$ and $\|\cdot\|_{2}$ are equivalent on $\mathbb{R}^n$, there exists a constant $C_{1}$ such that $ \|\mathbf{v}_n\|_{\infty} \leq C_{2} \| v\|_{H^2(\Omega)}, \ \forall v \in H^2(\Omega)$. Since $\mathbf{Q}$ is symmetric, its largest eigenvalue $\rho$ is non negative. Then we have:
$$
\renewcommand{\arraystretch}{2}
\begin{array}{lll}
\displaystyle G(u,v)&=&\mathbf{u}_n^t \mathbf{Q} \ \mathbf{v}_n + \lambda\int_{\Omega} (\Delta u)(\Delta v)\\
\displaystyle  &\leq& \rho \|\mathbf{u}_n\|_{\infty}\|\mathbf{v}_n\|_{\infty}+ \lambda |u|_{H^2(\Omega)}  |v|_{H^2(\Omega)} \\
&\leq &  \rho\ C_1^2  \|u\|_{H^2(\Omega)} \|v\|_{H^2(\Omega)} + \lambda\ C_0^2 \ \|u\|_{H^2(\Omega)} \|v\|_{H^2(\Omega)} \\
&\leq& \max{\left\{  \rho\ C_1^2 ,\lambda\ C_0^2\right\}} \|u\|_{H^2(\Omega)} \|v\|_{H^2(\Omega)}.
\end{array}
$$
And so the bilinear form $G(\cdot,\cdot)$ is also continuous.
\end{proof}

We are now ready to give the proof of Proposition  \ref{prop::weight_min_prob}.
\begin{proof}
First of all, given  $f\in H^2_{\mathbf{n}_0}$, the unique minimizer of the functional $\tilde{\mathcal{J}}_{\lambda}(\boldsymbol{\beta}, f)$ is given by:
\begin{equation}
\label{equ::express_betahat}
\tilde{\boldsymbol{\beta}}(f)= (\mathbf{X}^t\mathbf{W}\mathbf{X})^{-1}\mathbf{X}^t \mathbf{W}(\mathbf{z}-\mathbf{f}_n).
\end{equation}
To show that, we take the derivative of $\tilde{\mathcal{J}}_{\lambda}(\boldsymbol{\beta}, f)$ with respect to $\boldsymbol{\beta}$:
$$\frac{\partial \tilde{\mathcal{J}}_{\lambda}(\boldsymbol{\beta}, f)}{\partial \boldsymbol{\beta}}= -2 \mathbf{X}^t\mathbf{W}(\mathbf{z}-\mathbf{f}_n)+ (\mathbf{X}^t\mathbf{W}\mathbf{X})\boldsymbol{\beta}.$$
Since $\mathbf{X}$ is  a full-rank matrix and $\mathbf{W}$ is invertible (the $ii$th entry of $\mathbf{W}$ is in fact strictly positive, since it is  different from zero and $\geq 0$ by construction), $\mathbf{X}^t\mathbf{W}\mathbf{X}$ is invertible.  Finally the necessary condition $\partial \tilde{\mathcal{J}}_{\lambda}(\tilde{\boldsymbol{\beta}}, f)/ \partial \boldsymbol{\beta}=0$ is satisfied if and only if $\tilde{\boldsymbol{\beta}}$ is given by (\ref{equ::express_betahat}). Since for fixed $f$, $\tilde{\mathcal{J}}_{\lambda}(\boldsymbol{\beta}, f)$ is clearly convex, $\tilde{\boldsymbol{\beta}}$ is a minimum.

Now, plugging $\tilde{\boldsymbol{\beta}}$ into the objective function, we obtain the following form of the functional:
\begin{equation}
\nonumber
\tilde{\mathcal{J}}_{\lambda}(f)= \mathbf{z}^t \mathbf{Q}\ \mathbf{z}-2 \mathbf{f}_n\mathbf{Q}\ \mathbf{z}+{\mathbf{f}_n}^t \mathbf{Q}\ \mathbf{f}_n+\lambda \int_{\Omega} (\Delta f)^2.
\end{equation}
Since we want to optimize this functional with respect to  $f$ only, the problem becomes finding $\tilde{f} \in H^2_{\mathbf{n}_0}$  that minimizes:
\begin{equation}
\label{equ::rewrite_functional}
\mathcal{J}^*_{\lambda}(f)= {\mathbf{f}_n}^t \mathbf{Q}\ \mathbf{f}_n+\lambda \int_{\Omega} (\Delta f)^2 - 2 \mathbf{f}_n\mathbf{Q}\ \mathbf{z}.
\end{equation}
We can then write $\mathcal{J}^*_{\lambda}(f)=G(f,f)- 2F(f)$, where
$$G(u,v)=\mathbf{u}_n^t \mathbf{Q} \ \mathbf{v}_n + \lambda\int_{\Omega} (\Delta u)(\Delta v)\qquad \text{and}\qquad F(v)=\mathbf{v}_n^t \mathbf{Q} \ \mathbf{z}.$$
Lemma \ref{lem::coerc} ensures that  $G(\cdot,\cdot)$ is a coercive and continuous bilinear form on $H^2_{\mathbf{n}_0}\times H^2_{\mathbf{n}_0}$; moreover, $F$ is trivially a linear and continuous functional on $H^2_{\mathbf{n}_0}$.
Applying  the Lax-Milgram lemma and thanks to the symmetry of $G(\cdot,\cdot)$, the minimizer of the functional (\ref{equ::rewrite_functional}) is the function  $\tilde{f} \in H^2_{\mathbf{n}_0}$ such that
$$
\mathbf{u}_n^t\mathbf{Q}\ \tilde{\mathbf{f}}_n + \lambda \int_{\Omega} (\Delta u) (\Delta \tilde{f}) = \mathbf{u}_n^t\mathbf{Q}\ \mathbf{z}, \qquad \forall u\in H^2_{\mathbf{n}_0}
$$
We conclude that $\tilde{f}$ exists and is unique and hence also $\tilde{\boldsymbol{\beta}}$ exists and is unique.\end{proof}



\section{A general formulation of the model}
\label{app::linear_operator}

We here present a unified formulation of the model, that comprehends as special cases the model versions for geostatistical and for areal data presented in the paper.
Let $L:\mathcal{F}\to \mathbb{R}^n$ be a linear and continuous operator.
Assume that $Y_1,\dots,Y_n$ are independent, with $Y_i$  having a distribution within the exponential family, with mean $\mu_i$ and common scale paramenter $\phi$. Let  the  mean vector $\boldsymbol{\mu}$ be defined by: \begin{equation}
\nonumber
g(\boldsymbol{\mu}) = \boldsymbol{\theta}=  \boldsymbol{X}\boldsymbol{\beta} +\mathbf{f}_n \qquad \textrm{ where }  \mathbf{f}_n =L(f).
\end{equation}
We estimate  $\boldsymbol{\beta} \in \mathbb{R}^q$   and the spatial field $f\in\mathcal{F}$ by maximizing the penalized log-likelihood functional in (\ref{equ::pen_like}). We moreover define  $\boldsymbol{\Psi}$ as the $n\times K $ matrix whose $j-$th column is given by $L(\psi_j)$,  where $\psi_j$ is the $j-$th finite element basis. If $L$  is the linear operator that evaluates a function in $\mathcal{F}$ at $n$ spatial locations $\mathbf{p}_1,\dots,\mathbf{p}_n \in \Omega,$ then $\mathbf{f}_n =(f(\mathbf{p}_1),\dots,f(\mathbf{p}_n))^t$  and $    \boldsymbol{\Psi}$ has the form in (\ref{eq::Psi}); we obtain in this case the model version for geostatistical data. If $L$ is the linear  operator that returns the integrals of a function in $\mathcal{F}$  over $n$ disjoints subregions $D_1,\dots,D_n$  of  the domain $\Omega,$ then $\mathbf{f}_n = ( \int_{D_1} f , \dots, \int_{D_n} f)^t$ and the $n\times K $ matrix $\boldsymbol{\Psi}$ has entry $(i,j)$ given by $\int_{D_i}\psi_j;$  we obtain in this case the model version for areal data.

Owing to the linearity of the operator, the results presented  in Sections \ref{sec::just_PIRLS_variational} and \ref{sec::fem}
 and  in Appendices \ref{app::spatial-variation}, \ref{app::PIRLS} and \ref{app::proof_1}  carry over to the more general formulation of the model here presented, replacing the  definition of $\mathbf{f}_n $
 and $\boldsymbol{\Psi}$ with the more general definitions based on the linear operator $L$ given above.


\section{Test fields on horseshoe domain }
\label{app::test-fields}

\begin{figure}[bthp]
\centering
\begin{tikzpicture}[scale = 2.5]
\fill[color=gray!20] (3,0.9)-- (3,0.1) -- (3,0.1) arc (-90:90:0.4) -- cycle;
\fill[color=gray!20] (3,-0.9)-- (3,-0.1) -- (3,-0.1) arc (90:-90:0.4) -- cycle;
\fill[color=gray!20] (0,0.9) arc (90:270:0.9) -- (0,-0.9) -- (0,-0.1) arc (270:90:0.1)  -- cycle;

\draw (3.2,0.5) node[above]{$A$} ;
\draw (-0.6,0) node[above]{$B$} ;
\draw (3.2,-0.5) node[above]{$C$} ;

\draw[thick =2] (3,0.9)-- (0,0.9);
\draw[thick =2] (0,0.9) arc (90:270:0.9) ;
\draw[thick =2] (3,-0.9)-- (0,-0.9);
\draw[thick =2] (3,-0.9) arc (-90:90:0.4) ;
\draw[thick =2] (3,-0.1)-- (0,-0.1);
\draw[thick =2] (0,-0.1) arc (270:90:0.1) ;
\draw[thick =2] (3,0.1)-- (0,0.1);
\draw[thick =2] (3,0.1) arc (-90:90:0.4) ;

\draw[thick =2] (0,0.9) arc (90:270:0.9) ;

\draw[thick =2] (3,0.9)-- (0,0.9);
\draw[thick =2] (0,0.9) arc (90:270:0.9) ;
\draw[thick =2] (3,-0.9)-- (0,-0.9);

\draw[thick =2, dashed, color = blue] (3.4,0.5)-- (0,0.5);
\draw[thick =2, dashed, color = blue] (0,0.5) arc (90:270:0.5) ;
\draw[thick =2, dashed, color = blue] (3.4,-0.5)-- (0,-0.5);
\draw[->, thick =2, color= blue] (0,0) -- (-115:0.5) node[below] {$r$};
\draw[->, thick =2, color= black] (0,0) -- (-135:0.9) node[below left] {$2r-r_0$};

\draw[->, thick =2, color= black] (0,0) -- (135:0.1) node[above left] {$r_0$};
\draw[>=stealth,->, line width=1.3pt] (-1.3,0) -- (3.5,0) node[right] {$x$};
\draw[>=stealth,->, line width=1.3pt] (0,-1.3) -- (0,1.3) node[above] {$y$};
\draw[very thin,color=gray] (-1.3,-1.3) grid (3.5,1.3);
\draw plot[mark=|, mark size = 1pt] (1,-1) node[below right]{$1$};
\draw plot[mark=|, mark size = 1pt] (2,-1) node[below right]{$2$};
\draw plot[mark=|, mark size = 1pt] (3,-1) node[below right]{$3$};

\draw plot[mark=-, mark size = 1pt] (-1,0) node[above left]{$0$};
\draw plot[mark=-, mark size = 1pt] (-1,1) node[above left]{$1$};
\draw plot[mark=-, mark size = 1pt] (-1,-1) node[below left]{$-1$};
\end{tikzpicture}
\caption{\label{gr::horseshoedomain}Horseshoe domain used for the simulation studies of Section \ref{sec::sim}.}
\end{figure}
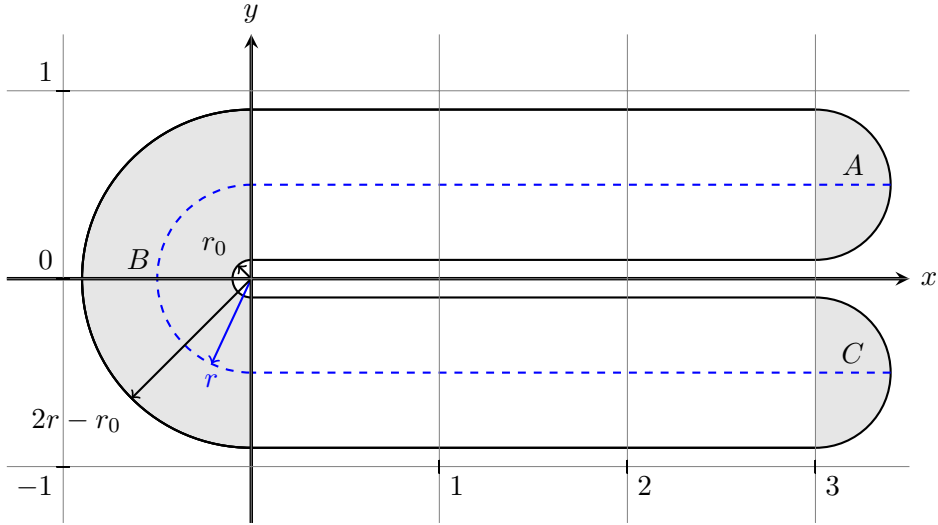

Figure \ref{gr::horseshoedomain} shows the Horseshoe domain used for the simulation studies of Section \ref{sec::sim}, where  $q =\frac{\pi r}{2}$ and   $\mathbf{1}_{\mathbf{A}}(x,y),$ $\mathbf{1}_{\mathbf{B}}(x,y)$ and $\mathbf{1}_{\mathbf{C}}(x,y)$ are indicator functions defined as
\begin{align*}
& \mathbf{1}_{\mathbf{A}}(x,y) =  \mathbf{1}_{\{(x-3)^2 + (y-r)^2 < (r-r_0)^2\}}(x,y)\mathbf{1}_{\{x>3\}}(x,y), \\
& \mathbf{1}_{\mathbf{B}}(x,y) = \mathbf{1}_{\{r_0^2 \leq x^2 + y^2 \leq (2r-r_0)^2\}}(x,y),\\
& \mathbf{1}_{\mathbf{C}}(x,y) = \mathbf{1}_{\{(x-3)^2 + (y+r)^2 < (r-r_0)^2\}}(x,y)\mathbf{1}_{\{x>3\}}(x,y). 
\end{align*}
We set $r=0.5$ and $r_0 = 0.1$.

The test field used for the simulation study with geostatistical data is defined as:
$$
f(x,y) =
\left\{
\renewcommand{\arraystretch}{2.5}
\begin{array}{ll}
\displaystyle -\frac{1}{8}\left[x + q + (y-r)^2 +10 \right] \mathbf{1}_{\mathbf{A}}(x,y) & x> 3,\quad r_0\leq y \leq 2r-r_0, \\
\displaystyle -\frac{1}{8}\left[x + q + (y-r)^2 +10 \right]  &0 \leq x\leq 3,\quad r_0\leq y \leq 2r-r_0, \\
\displaystyle  \frac{1}{8}\left[ \atan\left(\frac{y}{x}\right)r + (\sqrt{x^2 + y^2} -r )^2 +10\right]  \mathbf{1}_{\mathbf{B}}(x,y) & x<0, \\
\displaystyle  -\frac{1}{8}\left[ x + q + (y-r)^2 +10\right]  & 0 \leq x \leq 3,\quad -r_0\geq y \geq -2r+r_0, \\
\displaystyle  -\frac{1}{8}\left[ x + q + (y-r)^2 +10\right]  \mathbf{1}_{\mathbf{C}}(x,y)  & x > 3,\quad -r_0\geq y \geq -2r+r_0. \\
\end{array}
\right. $$
The test field used for the simulation study with areal data is defined as:
$$
f(x,y) =
\left\{
\renewcommand{\arraystretch}{2.5}
\begin{array}{ll}
\displaystyle \left[x + q + (y-r)^2 \right] \mathbf{1}_{\mathbf{A}}(x,y) & x> 3,\quad r_0\leq y \leq 2r-r_0, \\
\displaystyle \left[x + q + (y-r)^2 \right] &0 \leq x\leq 3,\quad r_0\leq y \leq 2r-r_0,\\

\displaystyle \left[- \atan\left(\frac{y}{x}\right)r + (\sqrt{x^2 + y^2} -r )^2\right]  \mathbf{1}_{\mathbf{B}}(x,y) & x<0, \\
\displaystyle \left[- x - q + (y-r)^2\right]   & 0 \leq x\leq 3,\quad -r_0\geq y \geq -2r+r_0, \\

\displaystyle \left[- x - q + (y-r)^2\right]  \mathbf{1}_{\mathbf{C}}(x,y)  & x > 3,\quad -r_0\geq y \geq -2r+r_0. \\
\end{array}
\right. $$
This is an affine transformation of the one considered for geostatistical data and coincides with the test function used in \cite{Wood08}.

\end{document}